\documentclass[letterpaper,twocolumn,10pt]{article}
\usepackage{usenix-2020-09}
\usepackage{tikz}
\usepackage{amsmath}
\usepackage{filecontents}
\usepackage{amsmath , amssymb , amsthm}
\usepackage{thmtools}
\usepackage{thm-restate}
\usepackage{scalerel}
\usepackage[utf8]{inputenc}
\usepackage{lipsum}
\usepackage[ruled, lined, linesnumbered, commentsnumbered, noend]{algorithm2e}
\usepackage[noend]{algpseudocode}
\usepackage{graphicx}
\usepackage{microtype} 
\DisableLigatures{encoding = *, family = * }
\usepackage[english]{babel}
\usepackage{blindtext}
\usepackage{url}
\usepackage{subcaption}
\usepackage{mathtools}
\usepackage{xcolor}
\usepackage{color, colortbl}
\usepackage[normalem]{ulem}
\usepackage{enumitem}
\usepackage{tablefootnote}
\setlist{nosep}

\colorlet{rowHighlight}{green!10}
\definecolor{takeawaycolor}{rgb}{0.75, 0, 0}
\definecolor{myred}{rgb}{0.75, 0, 0}
\definecolor{alg1}{rgb}{0.855, 1, 0.824}
\definecolor{alg2}{rgb}{1, 0.996, 0.824}
\definecolor{alg3}{rgb}{0.824, 0.914, 1}

\newcommand{\titlename}{Credence\xspace}
\newcommand{\name}{\textsc{Credence}\xspace}

\newcommand{\myitem}[1]{\vspace*{0.03in}\noindent\textbf{#1}}
\newcommand{\myitemit}[1]{\vspace*{0.02in}\noindent\textit{#1}}

\newcommand{\first}{\emph{(i)}\xspace}
\newcommand{\second}{\emph{(ii)}\xspace}
\newcommand{\third}{\emph{(iii)}\xspace}
\newcommand{\fourth}{\emph{(iv)}\xspace}
\newcommand{\fifth}{\emph{(v)}\xspace}

\newcommand{\ie}{i.e., \@}
\newcommand{\eg}{e.g., \@}

\newsavebox\thmbox
\theoremstyle{plain}
\newtheorem{theorem}{Theorem}
\newtheorem{observation}{Observation}

\newtheorem{definition}{Definition}

  {\begin{lrbox}{\thmbox}%
   \begin{minipage}{\dimexpr\linewidth-2\fboxsep}
   \myitem{Takeaway.}}%
  {\end{minipage}%
   \end{lrbox}%
   \begin{trivlist}
     \item[]\colorbox{lightgray}{\usebox\thmbox}
   \end{trivlist}}

\usepackage{titlesec}
\titlespacing*{\section}{0ex}{2ex plus .2ex minus .2ex}{1ex plus .2ex minus .2ex}
\titlespacing*{\subsection}{0ex}{1ex plus .2ex minus .2ex}{1ex plus .2ex minus .2ex}
\titlespacing*{\subsubsection}{0ex}{0.8ex plus .2ex minus .2ex}{0.5ex plus .2ex minus .2ex}

\begin{document}

\title{\titlename: Augmenting Datacenter Switch Buffer Sharing with ML Predictions
	\thanks{Author's version. Final version to appear in Usenix NSDI~2024.}
}

\author{
	{\rm Vamsi Addanki}\\
	TU Berlin
	\and
	{\rm Maciej Pacut}\\
	TU Berlin
	\and
	{\rm Stefan Schmid}\\
	TU Berlin
} 

\maketitle
\pagestyle{empty}
\thispagestyle{empty}

\begin{abstract}
	Packet buffers in datacenter switches are shared across all the switch ports in order to improve the overall throughput.
	The trend of shrinking buffer sizes in datacenter switches makes buffer sharing extremely challenging and a critical performance issue.
	Literature suggests that push-out buffer sharing algorithms have significantly better performance guarantees compared to drop-tail algorithms. Unfortunately, switches are unable to benefit from these algorithms due to lack of support for push-out operations in hardware.
	Our key observation is that drop-tail buffers can emulate push-out buffers if the future packet arrivals are known ahead of time. This suggests that augmenting drop-tail algorithms with predictions about the future arrivals has the potential to significantly improve performance.

	This paper is the first research attempt in this direction.
	We propose \name, a drop-tail buffer sharing algorithm augmented with machine-learned predictions.
	\name can unlock the performance only attainable by push-out algorithms so far. Its performance hinges on the accuracy of predictions.
	Specifically, \name
	achieves near-optimal performance of the best known push-out algorithm LQD (Longest Queue Drop) with perfect predictions, but \emph{gracefully} degrades to the performance of the simplest drop-tail algorithm Complete Sharing when the prediction error gets arbitrarily worse.
	Our evaluations show that \name improves throughput by $1.5$x compared to traditional approaches. In terms of flow completion times, we show that \name improves upon the state-of-the-art approaches by up to $95\%$ using off-the-shelf machine learning techniques that are also practical in today's hardware. We believe this work opens several interesting future work opportunities both in systems and theory that we discuss at the end of this paper.
\end{abstract}

\section{Introduction}
\label{sec:intro}
Datacenter switches come equipped with an on-chip packet buffer that is shared across all the device ports in order to improve the overall throughput and to reduce packet drops. Unfortunately, buffers have become increasingly expensive and chip-manufacturers are unable to scale up buffer sizes proportional to capacity increase~\cite{weiToN}. As a result, the buffer available per port per unit capacity of datacenter switches has been gradually reducing over time. Worse yet, datacenter traffic is bursty even at microsecond timescales~\cite{burstimc17,burstimc2022}. This makes it challenging for a buffer sharing algorithm to maximize throughput.
Recent measurement studies in large scale datacenters point-out the need for improved buffer sharing algorithms in order to reduce packet drops during congestion events~\cite{burstimc2022}.
To this end, buffer sharing under shallow buffers is an emerging critical problem in datacenters~\cite{abm,bai2023empowering}.

\begin{figure}[!t]
	\centering
	\includegraphics[width=1\linewidth, trim=0 1cm 0 3cm,clip]{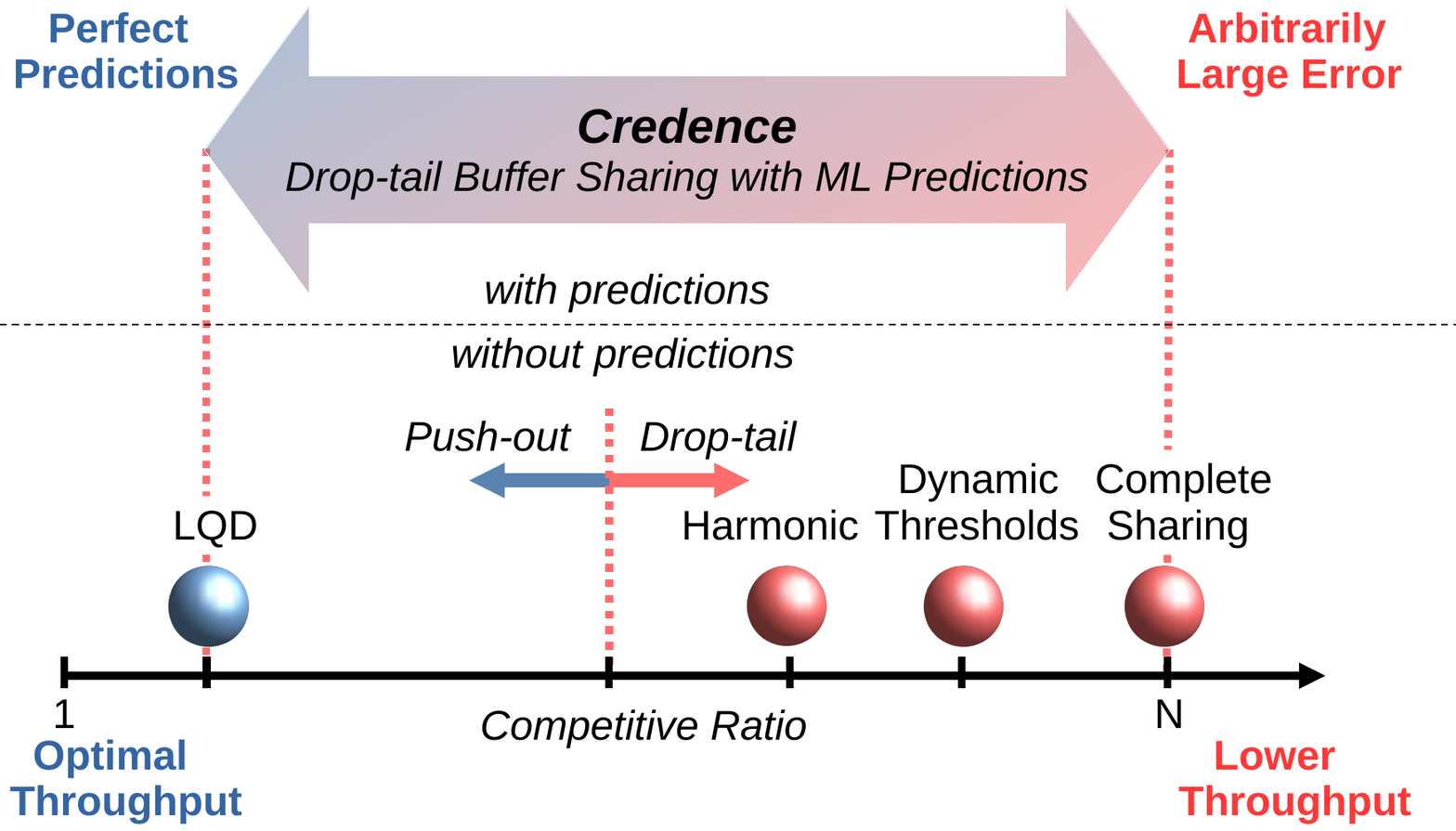}
	\caption{Augmenting drop-tail buffer sharing with ML predictions has the potential to significantly improve throughput compared to the best possible drop-tail algorithm (without predictions), and unlock the performance that was only attainable by push-out so far.}
	\label{fig:intro}
	\vspace{-5mm}
\end{figure}

The buffer sharing problem has been widely studied in the literature from an online perspective~\cite{borodin-book} with the objective to maximize throughput~\cite{competitiveBuffer,KESSELMAN2004161,breakingBarrier2,lqd2port,surveybufferonline}. Traditional online algorithms for buffer sharing can be classified into two types: \textbf{drop-tail} \eg Dynamic Thresholds (DT)~\cite{choudhury1998dynamic}, Harmonic~\cite{KESSELMAN2004161}, ABM~\cite{abm} and \textbf{push-out} \eg Longest Queue Drop (LQD). The performance gap of these algorithms compared to an offline optimal (clairvoyant) algorithm can be expressed in terms of the competitive ratio~\cite{borodin-book}. For instance, we say that an online algorithm is $2$-competitive if it performs at most $2$x worse compared to an offline optimal algorithm.
Figure~\ref{fig:intro} illustrates the performance spectrum of drop-tail and push out buffer sharing algorithms.
In terms of throughput-competitiveness, it is well-known that push-out algorithms perform significantly better than drop-tail algorithms. In fact, no deterministic drop-tail algorithm can perform better than a certain throughput (a~lower bound for competitive ratio), beyond which only push-out algorithms exist (Figure~\ref{fig:intro}). Table~\ref{table:c-ratio} presents the competitive ratios of known algorithms. Interestingly, LQD pushes out packets when the buffer is full, and it is $\approx 2$-competitive whereas Complete Sharing drops packets when the buffer is full, but it is $N+1$-competitive.

\begin{table}[t]
	\begin{center}
		\begin{tabular}{|c|c|c|c|}
			\hline
			\textbf{Algorithm}                                               & \textbf{Competitive Ratio} \\
			\hline
			Complete Sharing~\cite{competitiveBuffer}                        & $N+1$                      \\
			Dynamic Thresholds~\cite{competitiveBuffer,choudhury1998dynamic} & $\mathcal{O}(N)$           \\
			Harmonic~\cite{KESSELMAN2004161}                                 & $\ln(N)+2$                 \\
			LQD (push-out)~\cite{competitiveBuffer,breakingBarrier2}         & $1.707$                    \\
			LateQD (clairvoyant)~\cite{DBLP:journals/corr/abs-1907-04399}    & $1$                        \\
			\hline
			\rowcolor{rowHighlight}
			\name                                                            & $\min(1.707\ \eta,\ N)$    \\
			\hline
		\end{tabular}
	\end{center}
	\vspace{-2mm}
	\caption{\name's performance smoothly depends on the prediction error ($\eta$). \name outperforms traditional drop-tail buffer sharing algorithms and performs as good as push-out when the predictions are perfect ($\eta=1$) but is also never worse than Complete Sharing even when the predictions are bad ($\eta\rightarrow \infty$). $N$ denotes the number of ports.}
	\label{table:c-ratio}
	\vspace{-5mm}
\end{table}

Intuitively, the poor throughput-competitiveness of drop-tail buffers owes it to the fundamental challenge that utilizing the buffer for some queues comes at the cost of deprivation of buffer for others~\cite{abm}. To this end, drop-tail algorithms proactively drop packets \ie packets are dropped even when the buffer has remaining space~\cite{competitiveBuffer,choudhury1998dynamic,trafficaware,fab,KESSELMAN2004161}. On one hand, maintaining remaining buffer space is necessary to serve future packet arrivals. On the other hand, maintaining remaining buffer space could lead to under-utilization, throughput loss and excessive packet drops. In contrast, the superior throughput-competitiveness of push-out algorithms owes it to their fundamental advantage to push out packets instead of dropping them. Hence, push-out algorithms can utilize the entire buffer as needed and only push out packets when multiple ports contend for buffer space. Although push-out algorithms offer far superior performance guarantees compared to drop-tail, hardly any datacenter switch supports push-out operations for the on-chip shared buffer. This begs the question: Are drop-tail buffer sharing algorithms ready for the trend of shrinking buffer sizes?

Our key observation is that every push-out algorithm can be converted to a drop-tail algorithm. However, such a conversion requires certain (limited) visibility into the future packet arrivals. Specifically, pushing out a packet is equivalent to dropping the packet when it arrives. Recent advancements in traffic predictions play a pivotal role in providing such visibility into the future packet arrivals: paving a way for better drop-tail buffer sharing algorithms.

In this paper, we take the first step in this direction. Figure~\ref{fig:intro} illustrates our perspective. We propose \name, a drop-tail buffer sharing algorithm augmented with machine-learned predictions. \name's performance is tied to the accuracy of these predictions. As the prediction error decreases, \name unlocks the performance of push-out algorithms and reaches the performance of the best-known algorithm. Even when the prediction error grows arbitrarily large, \name offers at least the performance of the simplest drop-tail algorithm Complete Sharing. Table~\ref{table:c-ratio} gives the competitive ratio of \name as a function of the prediction error $\eta$. Importantly, \name's performance \emph{smoothly} varies with the prediction error, generalizing the performance space between the known push-out and drop-tail algorithms.
Hence, \name achieves the three goals of prediction-augmented algorithms, in the literature referred to as consistency, robustness and smoothness~\cite{10.1145/3528087, NEURIPS2018ML}.

In addition to the theoretical guarantees for \name's performance, our goal is also its practicality. Specifically, without predictions, \name's core logic only uses additions, subtractions, and does not add additional complexity compared to existing approaches. For predictions, we currently use random forests, which have been recently shown to be feasible on programmable switches at line rate~\cite{busse2019pforest,10229100}. A full implementation of \name in hardware unfortunately requires switch vendor involvement since buffer sharing is merely a blackbox even in programmable switches. With this paper, we wish to gain attention from switch vendors on the fundamental blocks required for such algorithms to be deployed in the dataplane.
We currently implement \name in NS3 to evaluate its performance using realistic datacenter workloads.
We present a detailed discussion on the practicality of \name later in this paper.

Our evaluations show that \name performs $1.5$x better in terms of throughput and up to $95$\% better in terms of flow completion times, compared to alternative approaches.

We believe \name is a stepping stone towards further improving buffer sharing algorithms. Especially, achieving better performance than \name under large prediction error remains an interesting open question. Our approach of augmenting buffer sharing with predictions is not limited to drop-tail algorithms, but push-out algorithms can also benefit from predictions. We discuss exciting future research directions both in systems and theory at the end of this paper.

\noindent In summary, our key contributions in this paper are:
\begin{itemize}[label=\small{\textcolor{takeawaycolor}{$\blacksquare$}}, topsep=0pt, itemsep=0pt,leftmargin=*]
	\item \name, the first buffer sharing algorithm augmented with predictions, achieving near-optimal performance with perfect predictions while also guaranteeing performance under arbitrarily large prediction error, and gradually degrading the performance as the prediction error increases.
	\item Extensive evaluations using realistic datacenter workloads, showing that \name outperforms existing approaches in terms of flow completion times.
	\item  All our artifacts have been made publicly available at \url{https://github.com/inet-tub/ns3-datacenter}.
\end{itemize}

\section{Motivation}
\label{sec:motivation}
In this section, we provide a brief background and motivate our approach by highlighting the drawbacks of traditional approaches. We show the potential for reaching close-to-optimal performance when buffer sharing algorithms are augmented with machine-learned predictions. To this end, we first describe our model and throughput competitiveness (\S\ref{sec:model-competitive}). We then discuss the drawbacks of existing approaches (\S\ref{sub:drawbacks}). We show that a renewed hope for improved buffer sharing is enabled by the recent rise in algorithms with predictions~(\S\ref{sub:hope}).

\subsection{Buffer Sharing from Online Perspective}\label{sec:model-competitive}
A network switch receives packets one after the other at each of its ports. The switch does not know the packet arrivals ahead of time. This makes buffer sharing inherently an \emph{online} problem \ie algorithms must take instantaneous decisions upon packet arrivals without the knowledge of the future. In order to systematically understand the performance of such algorithms, we take an online approach following the classical model in the literature~\cite{competitiveBuffer,KESSELMAN2004161,lqd2port,breakingBarrier2,surveybufferonline}. In this section, we describe our model intuitively, and we refer to Appendix~\ref{app:model} for formal definitions. Figure~\ref{fig:switch} illustrates the model.

\medskip
\myitem{Buffer model:} We consider an output-queued switch with $N$ ports and a buffer size of $B$. Buffer is shared across all the ports. A buffer sharing algorithm takes buffering decisions that we describe next.
We assume that time is discrete. At most $N$ packets can arrive in a single timeslot (since there are $N$ ports), and each port removes at most one packet in a timeslot.

\medskip
\myitem{Online algorithm:}
When a packet arrives, a buffer sharing algorithm determines whether it should be accepted into the available buffer space. Drop-tail algorithms can only accept or discard incoming packets, while push-out algorithms can also remove packets from the buffer.

\medskip
\myitem{Objective:} The network throughput is of  utmost importance for datacenter operators since throughput often relates to the cost in typical business models (\eg $\$$ per bandwidth usage). We hence consider throughput as an objective function, following the literature. Specifically, for any packet arrival sequence, our objective is to maximize the total number of transmitted packets.
The throughput maximization objective is closely related to packet drops minimization objective.
In this sense, our objective captures two important performance metrics \ie throughput and packet drops.

\medskip
\myitem{Competitive Ratio:} We use competitive ratio as a measure to compare the performance of an online algorithm to the optimal offline algorithm.
Specifically, let $ALG$ and $OPT$ be an online and optimal offline algorithm correspondingly. Let $ALG(\sigma)$ be the throughput of $ALG$ for the packet arrival sequence $\sigma$. We say an algorithm $ALG$ is $c$-competitive if the following relation holds for any packet arrival sequence.
\[
	OPT(\sigma) \le c\cdot ALG(\sigma)
\]

Competitive ratio is a particularly interesting metric for buffer sharing since it offers performance guarantees without any assumptions on specific traffic patterns. For example, the buffer may face excessive packet drops or may temporarily experience throughput loss due to bursty traffic. One could argue that the buffer sharing algorithm is the culprit and should have allocated more buffer to the bursty traffic. While this may have solved the problem for a particular bursty arrival, the same solution could result in unexpected drops and throughput loss if there were excessive bursty arrivals \ie large bursts could monopolize the buffer. Instead, from an online perspective, better competitive ratio indicates that the buffer sharing algorithm performs close to optimal under any traffic conditions.

\medskip
\noindent{\textcolor{takeawaycolor}{$\blacksquare$ \textbf{\textit{Takeaway.}}} \textit{A buffer sharing algorithm with lower competitive ratio improves the throughput of the switch and reduces packet drops under worst-case packet arrival patterns.}}

\begin{figure}[t]
	\centering
	\includegraphics[width=0.8\linewidth]{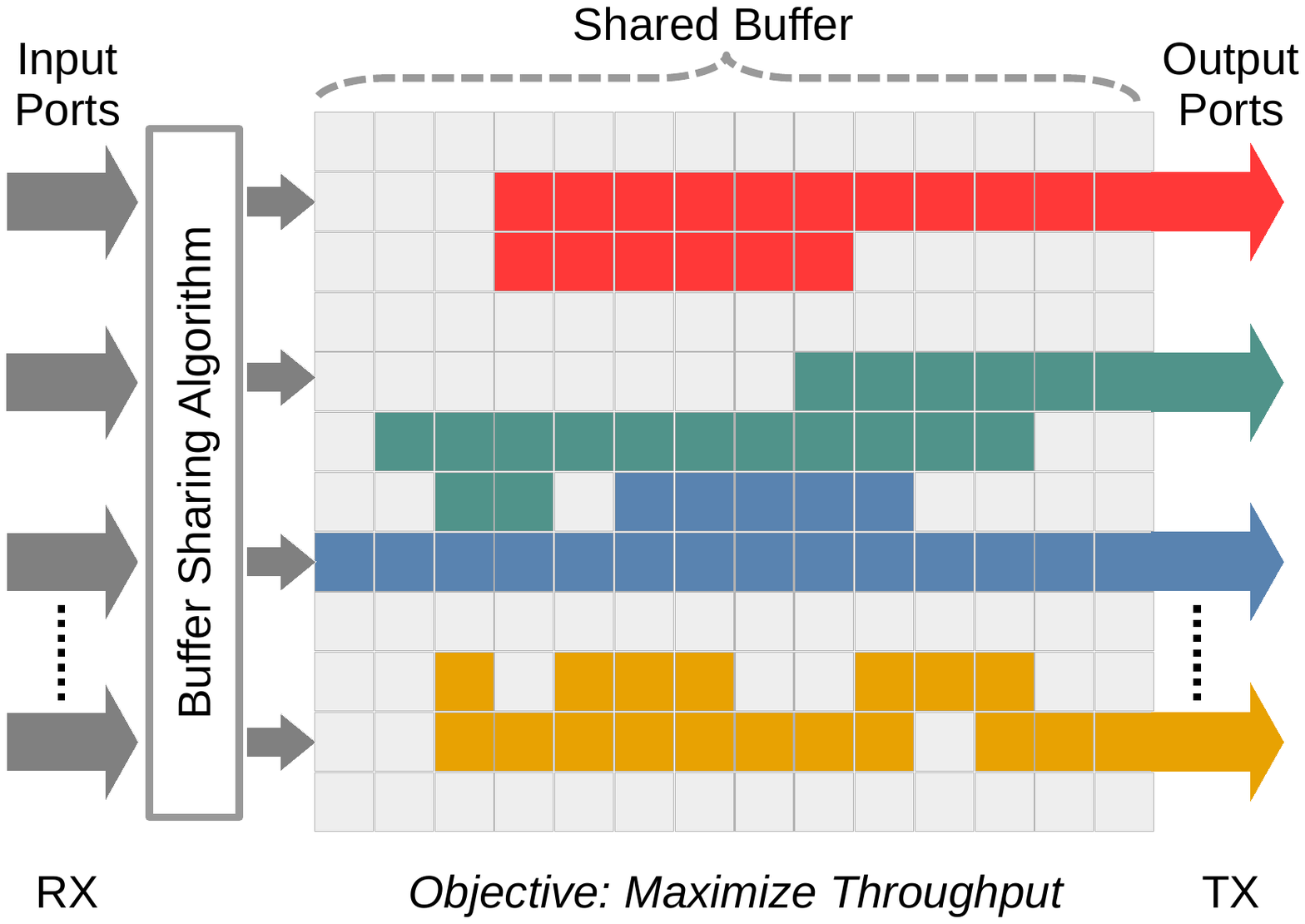}
	\caption{The switch has a buffer size of $B$ shared across $N$ output ports. Each color indicates the packets residing in the shared buffer corresponding to each port. A buffer sharing algorithm takes decisions (accept or drop) for each input packet.}
	\label{fig:switch}
	\vspace{-3mm}
\end{figure}

\subsection{Drawbacks of Traditional Approaches}
\label{sub:drawbacks}
We observe two main drawbacks of traditional buffer sharing algorithms,
both affecting the competitive ratio. First, algorithms proactively and unnecessarily drop packets in view of accommodating future packet arrivals. Second, algorithms reactively drop packets when the buffer is full and incur throughput loss, which could have been avoided.
We argue that these drawbacks are rather fundamental to drop-tail algorithms and cannot be addressed by traditional online approaches.

\begin{figure*}
	\begin{minipage}[b]{0.49\linewidth}
		\centering
		\begin{subfigure}[b]{0.495\linewidth}
			\centering
			\includegraphics[trim=4cm 0.5cm 0 0.5cm,clip,width=1\linewidth]{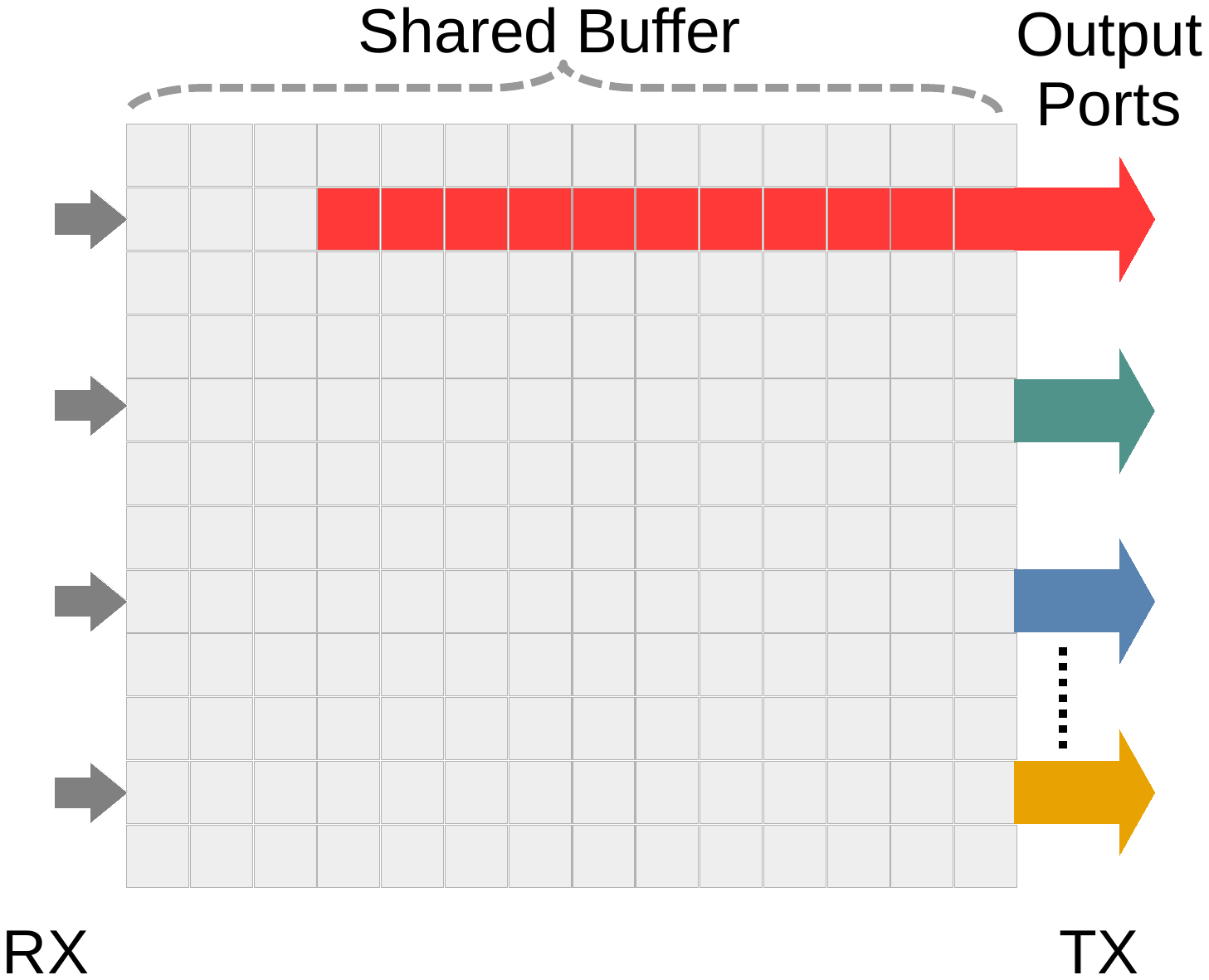}
			\caption{ALG}
			\label{fig:example1-alg}
		\end{subfigure}
		\begin{subfigure}[b]{0.495\linewidth}
			\centering
			\includegraphics[trim=4cm 0.5cm 0 0.5cm,clip,width=1\linewidth]{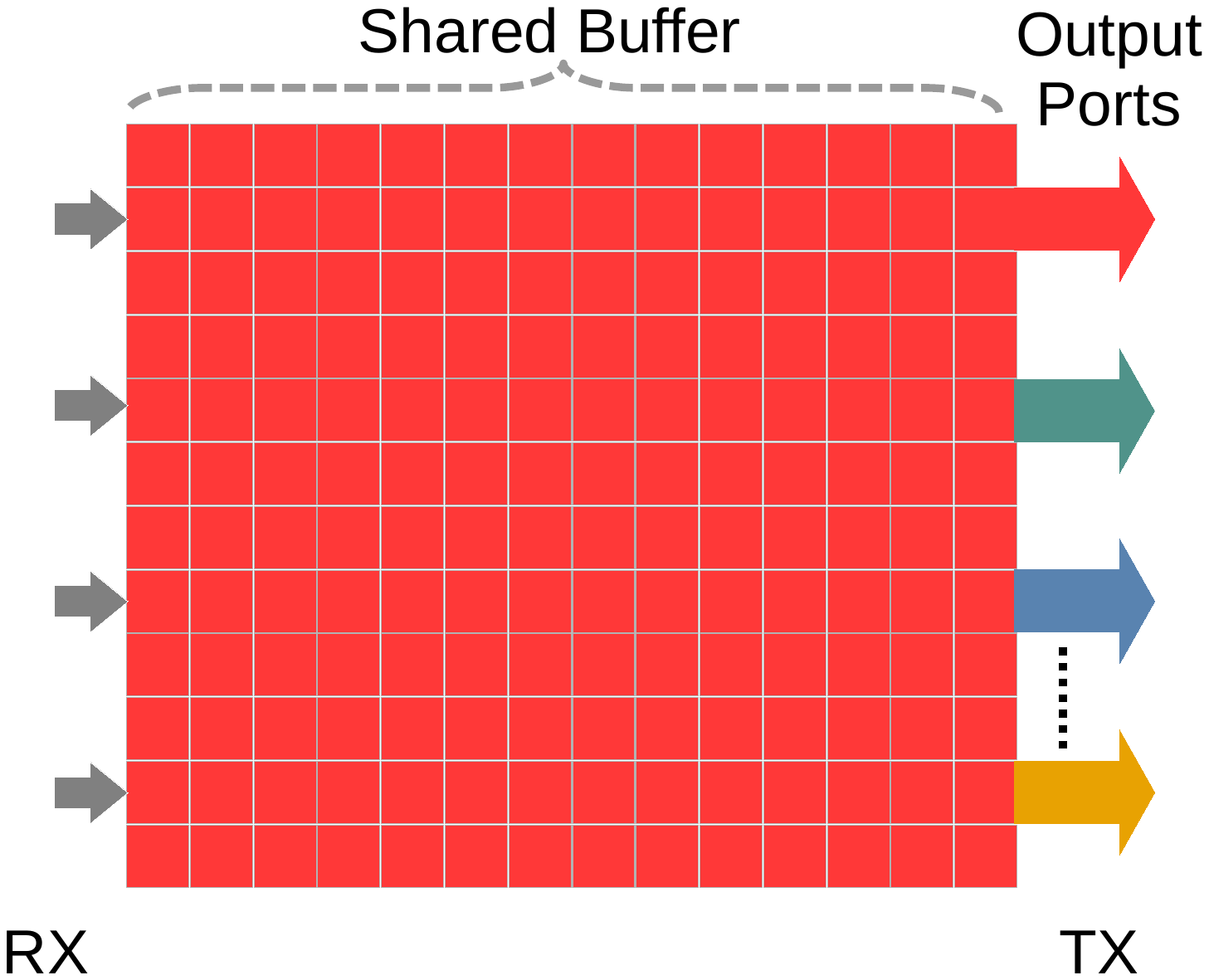}
			\caption{OPT}
			\label{fig:example1-opt}
		\end{subfigure}
		\caption{Upon a large burst arrival, a typical drop-tail algorithm (ALG) \emph{proactively} drops the incoming packets in anticipation of future bursts and significantly under-utilizes the buffer. In this case, an optimal offline algorithm accepts the entire burst without any packet drops.}
		\label{fig:example1}
	\end{minipage}\hfill
	\begin{minipage}[b]{0.49\linewidth}
		\centering
		\begin{subfigure}[b]{0.495\linewidth}
			\centering
			\includegraphics[trim=4cm 0.5cm 0 0.5cm,clip,width=1\linewidth]{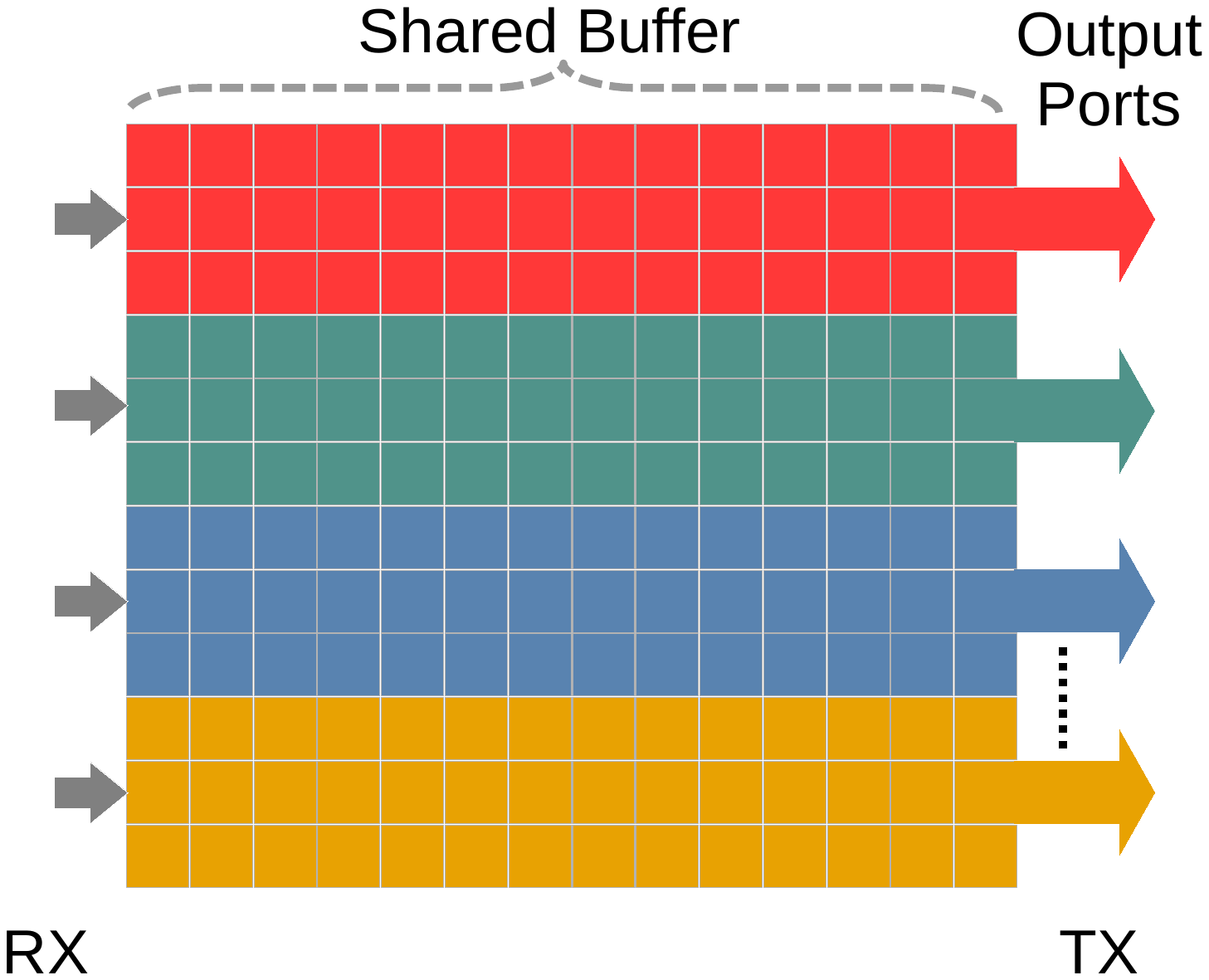}
			\caption{ALG}
			\label{fig:example2-alg}
		\end{subfigure}
		\begin{subfigure}[b]{0.495\linewidth}
			\centering
			\includegraphics[trim=4cm 0.5cm 0 0.5cm,clip,width=1\linewidth]{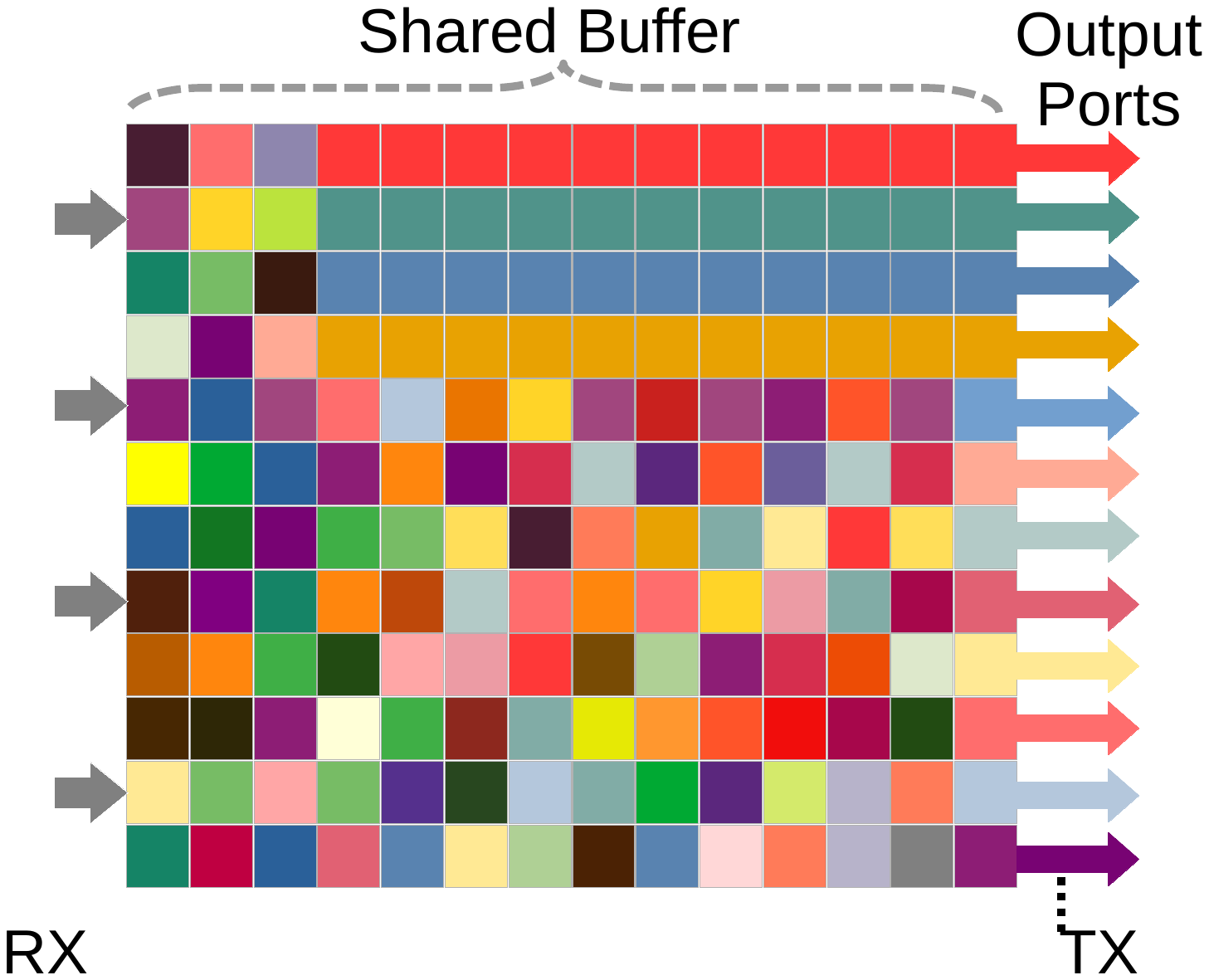}
			\caption{OPT}
			\label{fig:example2-opt}
		\end{subfigure}
		\caption{In pursuit of high burst absorption, a drop-tail algorithm ALG may absorb bursts but this results in excessive \emph{reactive} drops for the future packet arrivals. In this case, an optimal offline algorithm OPT drops few packets such that the overall throughput is maximized.}
		\label{fig:example2}
	\end{minipage}
	\vspace{-4mm}
\end{figure*}

\medskip
\myitem{Proactive unnecessary packet drops $\rightarrow$ throughput loss:}
A drop-tail buffer sharing algorithm typically drops packets even if there is remaining buffer space available~\cite{choudhury1998dynamic,fab,abm}. We refer to such drops as proactive drops. Being proactive is indeed necessary in order to accommodate transient bursts. However, proactive packet drops and the corresponding remaining buffer space ends up being wasteful if the future packet arrivals do not need additional buffer space (if the anticipated burst does not arrive). Figure~\ref{fig:example1-alg} and Figure~\ref{fig:example1-opt} illustrate an example. Consider a traffic pattern where there is little to no congestion on all the ports but once in a while, a large burst appears. Specifically, the buffer is empty initially and a large burst of size $B$ appears. A deterministic drop-tail algorithm has two choices: \first accept a portion of the burst and proactively drop the rest of the burst or \second accept the entire burst. Typical algorithms in the literature choose the former in view of accommodating future packet arrivals. An optimal offline algorithm that knows the arrivals ahead of time would accept the entire burst of size $B$ in this case. This makes an online algorithm at least $c$-competitive for this particular arrival pattern, where $\frac{1}{c}$ is the fraction of the burst accepted: since the optimal solution accepts and transmits $B$ packets over time, whereas an online algorithm only accepts and transmits only $\frac{B}{c}$ packets over time. We observe that recent works focus on minimizing proactive unnecessary packet drops by prioritizing bursty traffic to the extent that they allow burst on a single port to monopolize the buffer~\cite{fab,abm,trafficaware}.
However, note that competitive ratio is not defined for a particular arrival sequence, but over all scenarios. To this end, accepting a larger burst size may be helpful in the above example but if there were indeed future packet arrivals on other ports that need buffer, the algorithm incurs excessive reactive drops (described next) and throughput loss.

\myitem{Reactive avoidable packet drops $\rightarrow$ throughput loss:} Any drop-tail algorithm is forced to drop the incoming packets once the shared buffer is full. We call such drops reactive drops. Reactive drops result in throughput loss if the algorithm fills up significant portion of the buffer on a small set of ports but reactively drops incoming packets to other ports. Figure~\ref{fig:example2-alg} and Figure~\ref{fig:example2-opt} illustrate an example. Consider that the buffer is initially empty and four simultaneous bursts each of size $B$ arrive to four ports. If an algorithm proactively drops a significant portion of the bursts, it would suffer under arrival sequences such as in the previous example (Figure~\ref{fig:example1-alg}). Alternatively, the algorithm may choose to accept a larger portion of the bursts and ends up filling up the entire buffer in aggregate. At this point, several short bursts arrive to multiple other ports. An optimal offline algorithm accepts only a fraction of the large bursts such that it is able to accommodate upcoming short bursts. In doing so, the optimal algorithm benefits in throughput since the switch transmits packets from more number of ports. However, since the online algorithm fills up the entire buffer due to the initial large bursts, it is forced to reactively drop the upcoming short bursts, losing throughput. In fact, a similar arrival pattern for Dynamic Thresholds yields at least $\Omega\left(\sqrt{\frac{N}{\log(N)}}\right)$-competitiveness~\cite{competitiveBuffer}. The known upper bound for Dynamic Thresholds is $\mathcal{O}(N)$~\cite{competitiveBuffer}. Further, it has been shown in the literature that no deterministic drop-tail algorithm can be better than $\Omega\left(\frac{\log(N)}{\log(\log(N))}\right)$-competitive~\cite{KESSELMAN2004161}.

Interestingly, push-out algorithms are not prone to the problems discussed above, since they can take revocable decisions \ie to accept a packet and drop it later. Hence, push-out algorithms do not have to maintain free space in the buffer in order to accommodate transient bursts. Instead, such algorithms can defer the dropping decision until the moment the drop turns out to be necessary.

\medskip
\noindent{\textcolor{takeawaycolor}{$\blacksquare$ \textbf{\textit{Takeaway.}}} \textit{Traditional drop-tail algorithms are fundamentally limited in throughput-competitiveness as they are unable to effectively navigate proactive and reactive drops due to the online nature of the problem \ie future packet arrivals are unknown to the algorithm.}}

\subsection{Predictions: A Hope for Competitiveness}
\label{sub:hope}

Given that the fundamental barrier in improving drop-tail buffer sharing algorithms is the lack of visibility into the future arrivals, we turn towards predictions. The recent rise of algorithms with predictions offers a renewed hope for competitive buffer sharing.
Algorithms with predictions successfully enabled close to optimal performance for various classic problems~\cite{NEURIPS2018ML}. The core idea is to guide the underlying online algorithm with certain knowledge about the future obtained via predictions. The machine-learned oracle that produces predictions is considered a blackbox with a certain error. The main challenge is to offer performance guarantees at the extremes \ie close to optimal performance under perfect predictions and a minimum performance guarantee when the prediction error gets arbitrarily large. Further, it is desirable that the competitiveness of the algorithm \emph{smoothly} degrades as the prediction error grows.

\subsubsection{Prediction Model}\label{sec:prediction-model}
In the context of the buffer sharing problem, there are several prediction models that can be considered \eg drops or packet arrivals. In this paper, we assume that a blackbox machine-learned oracle predicts packet drops. Our choice is due to the fact that packet drops are the basic decisions made by an algorithm. Concretely, we consider an oracle that predicts whether an incoming packet would eventually be dropped (or pushed out) by the Longest Queue Drop (LQD) algorithm serving the same packet arrival sequence. We classify the predictions into four types: \first true positive \ie a correct prediction that a packet is eventually dropped by LQD, \second false negative \ie an incorrect prediction that a packet is eventually transmitted by LQD, \third false positive \ie an incorrect prediction that a packet is eventually dropped by LQD and \fourth true negative \ie a correct prediction that a packet is eventually transmitted by LQD. Figure~\ref{fig:predictions} summarizes this classification. Following the literature~\cite{NEURIPS2018ML,10.1145/3528087}, our goals for prediction-augmented buffer sharing are consistency, robustness and smoothness.

\medskip
\noindent \textbf{$\alpha$-Consistent} buffer sharing algorithm has a competitive ratio~$\alpha$ when the predictions are all true \ie perfect predictions.

\medskip
\noindent \textbf{$\beta$-Robust} buffer sharing algorithm has a competitive ratio $\beta$ when the predictions are all false \ie large prediction error.

\medskip
\noindent \textbf{Smoothness} is a desirable property such that the competitive ratio degrades smoothly as the prediction error grows \ie a small change in error does not drastically influence the competitive ratio.

Our goal is to design a prediction-augmented buffer sharing algorithm that is close to $1$-consistent (with perfect predictions) \ie near-optimal, at most $N$-robust (with arbitrarily large error) \ie not worse than Complete Sharing algorithm, and has the desirable property of smoothness.

\subsubsection{Common Pitfalls}

It is intuitive that predictions can potentially improve the performance of a drop-tail algorithm. For instance, in the examples from Figure~\ref{fig:example1} and Figure~\ref{fig:example2}, our prediction-augmented online algorithm could take nearly the same decisions as a push-out algorithm. However, the main challenge is to ensure robustness and smoothness. If an algorithm blindly trusts the predictions, we observe that false positive and false negative predictions have a significantly different impact on the performance.

\medskip
\myitem{Excessive false positives can lead to starvation:} The worst case for a naive algorithm that blindly trusts predictions is when all the predictions are false positives. In this case, the algorithm ends up dropping every incoming packet. Blindly trusting false predictions could lead to a competitive ratio worse than the simplest drop-tail algorithm Complete Sharing \ie the competitive ratio becomes unbounded ($\infty$-robust) if the predictions are mostly false positives.

\medskip
\myitem{A single false negative can hurt throughput forever:} A naive algorithm that blindly relies on false negative predictions is susceptible to adverse effects that propagate over time. Consider a packet arrival sequence that hits only one queue initially and consider that the predictions are all true negatives until the queue length reaches $B-1$, where $B$ is the total buffer size. At this point, one more packet arrives and our prediction is a false negative. As a result, our naive algorithm has a queue of size $B$ and the optimal algorithm has a queue of size $B-1$. Note that all non-empty queues drain one packet after each timeslot. From here on, in every timeslot, one packet (first) arrives to the large queue and one packet (second) arrives to any other queue. Also consider that all the predictions are true from now on. The optimal algorithm accepts both first and second packet in every timeslot. However, in every timeslot our naive approach can only accept the first packet to the large queue and cannot accept the second packet since the buffer is full. Notice that relying on just one false negative resulted in cumulative drops in this case even though all other predictions were true. In fact, a tiny error such as just $N$ number of false negatives even with all other predictions being true could result in a competitive ratio for a naive approach as worse as Complete Sharing.

\begin{figure}[t]
	\centering
	\includegraphics[width=0.7\linewidth]{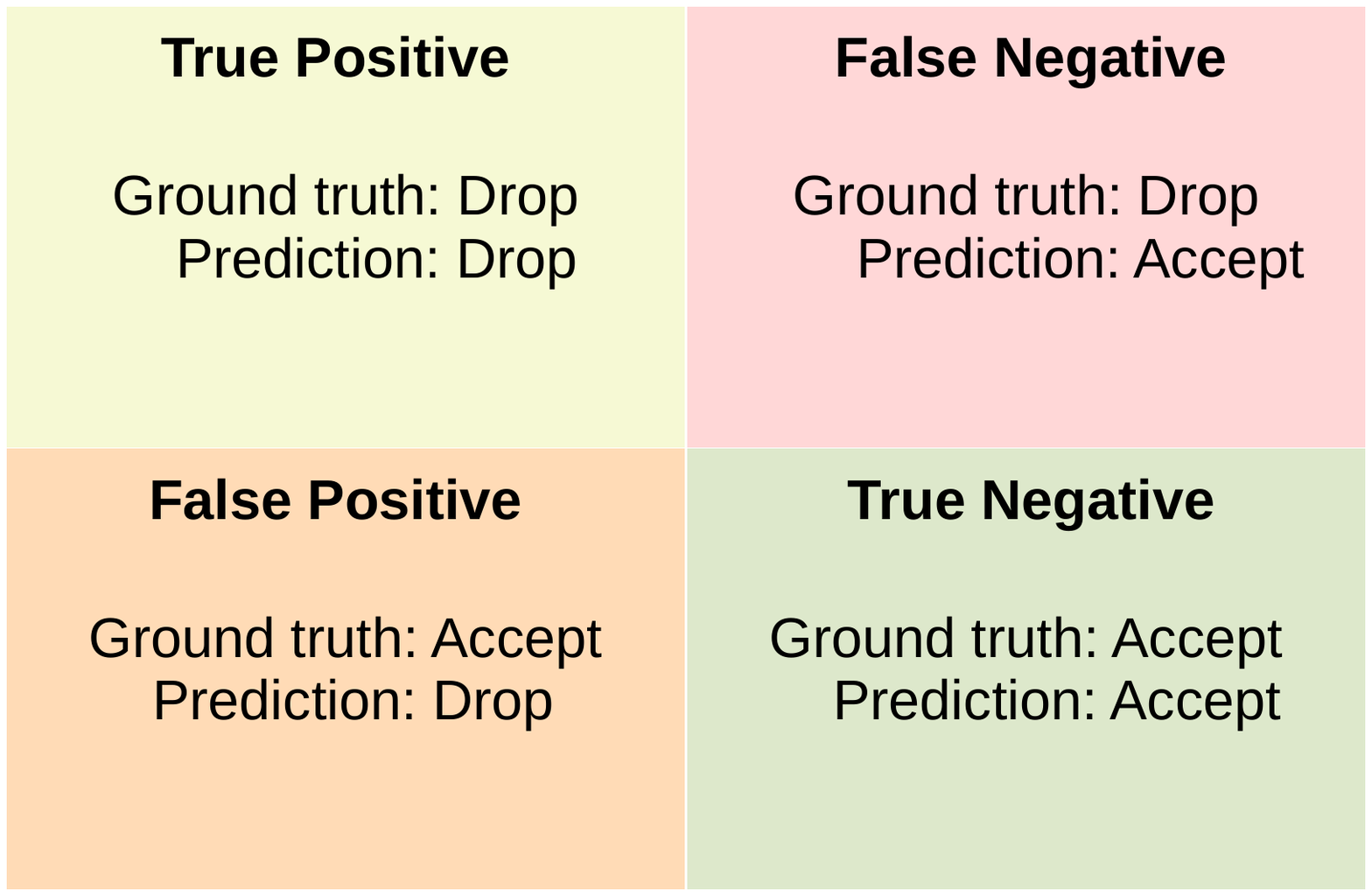}
	\vspace{-3mm}
	\caption{Confusion matrix for our prediction model.}
	\label{fig:predictions}
	\vspace{-3mm}
\end{figure}

\medskip
\noindent{\textcolor{takeawaycolor}{$\blacksquare$ \textbf{\textit{Takeaway.}}} \textit{Augmenting drop-tail algorithms with predictions has the potential to unlock the optimal performance. Ensuring performance guarantees with inaccurate predictions remains a challenge.}}

\section{Prediction-Augmented Buffer Sharing}\label{sec:credence}
Reflecting on our observations in \S\ref{sec:motivation}, our goal is to design a drop-tail buffer sharing algorithm that performs close to optimal with perfect predictions but also provides a minimum performance guarantee when the prediction error is arbitrarily large. In essence, our aim is to enable performance improvement in terms of throughput and packet drops in datacenter switches. To this end, we first present an overview of our algorithm (\S\ref{sec:overview}). We then present the workings of \name (\S\ref{sec:algorithm}) and discuss its properties (\S\ref{sec:properties}). Finally, we discuss the practicality of \name (\S\ref{sec:practicality}).

\subsection{Overview}
\label{sec:overview}
In a nutshell, \name relies on predictions and \emph{follows} a push-out algorithm, reaching close to optimal performance under perfect predictions. \name cleverly takes certain decisions independent of the predictions in order to guarantee a minimum performance. Further, \name's competitiveness gradually degrades as prediction error grows (a property known as smoothness~\cite{10.1145/3528087}), hence the algorithm still performs near-optimally when predictions are slightly inaccurate.

\myitem{\name \emph{follows} Longest Queue Drop algorithm:} Our design of \name consists of two key ingredients. First, \name uses thresholds as a drop condition irrespective of the predictions. \name treats thresholds as queue lengths of LQD and updates the thresholds based on the LQD algorithm (simply arithmetic) upon every packet arrival. Second, \name relies on predictions as long as the queue lengths satisfy the corresponding thresholds. The combination of thresholds and predictions allows \name to closely \emph{follow} the Longest Queue Drop algorithm (LQD) without requiring push-out operations\footnote{Recall that LQD is close to optimal with a competitive ratio of $1.707$.}.

\myitem{\name guarantees performance under extremities:}
When all the predictions are perfectly accurate, \name achieves a competitive ratio of $1.707$ (consistency) due to the straight-forward argument that the drops by \name and LQD are equivalent for true predictions.
In order to guarantee a minimum performance even with arbitrarily large prediction error (robustness), \name bypasses the threshold and predictions as long as the longest queue is within $\frac{B}{N}$ size. Here, $B$ is the buffer size and $N$ is the number of ports. This allows \name to be most $N$-competitive even under large prediction error, similar to the Complete Sharing algorithm. We prove our claim formally in Appendix~\ref{app:predictions}.

\myitem{\name smoothly degrades with prediction error:} We design our error function in terms of the performance of LQD and the predicted drops. We analyze the types of drops incurred by \name due to false positive and false negative predictions. This allows us to show that \name satisfies the smoothness property \ie the competitive ratio smoothly degrades from $1.707$ to $N$ as the prediction error~grows.

\subsection{\name}
\label{sec:algorithm}
We now present \name and explain how it operates. Algorithm~\ref{alg:follow-pred} presents the pseudocode of \name. Our pseudocode is simplified to discrete time for ease of presentation and for simplicity of analysis. It can be trivially extended to continuous time, and our implementation incorporates it\footnote{Our source code will be made publicly available online.}.

\myitem{Arrival:} Upon a packet arrival, \name has three important steps that are highlighted in Algorithm~\ref{alg:follow-pred}.
First, \name updates the threshold for the current queue (highlighted in blue).
Second, \name takes a decision based on the thresholds and predictions whether or not to accept the incoming packet (highlighted in yellow). Finally, the packet is either accepted or dropped. We next describe each of these steps in detail.
Third, depending on the state of the buffer, \name bypasses the thresholds and predictions with a safeguard condition in order to accept or drop the incoming packet (highlighted in green).

\myitem{Thresholds:} \name updates its thresholds based on the longest queue drop algorithm. Specifically, upon a packet arrival at time $t$ to a queue $i$, \name increments the threshold $T_i(t)$ for queue $i$ by the packet size. If upon arrival the sum of thresholds $\Gamma(t)$ is equal to the buffer size $B$, then \name first decrements the longest queue threshold by packet size and then increments the threshold for queue $i$ by the packet size. Note that upon a packet arrival to a queue, the corresponding threshold is updated before accepting or dropping the packet.

\myitem{Drop criterion:} Similar to existing threshold-based algorithms, \name also uses thresholds as a drop criterion. \name compares the queue length $q_i(t)$ of a queue $i$ against its threshold $T_i(t)$ and drops an incoming packet if the queue length is larger than or equal to the corresponding threshold. If and only if an incoming packet satisfies the thresholds, then \name takes input from a machine-learned oracle that predicts whether to accept or drop according to our prediction model discussed in \S\ref{sec:prediction-model}. Finally, based on the thresholds and predictions, \name either accepts or drop the incoming packet.

\myitem{Safeguard:} In order to bound \name's competitiveness under arbitrarily large prediction error, we bypass the above drop criterion under certain cases. Specifically, when the longest queue length is less than $\frac{B}{N}$, \name always accepts a packet irrespective of the thresholds and predictions. This ensures that \name is at least $N$-competitive even with large prediction error. Our safeguard is based on the observation that even the push-out longest queue drop algorithm cannot push out a packet from a queue less than $\frac{B}{N}$ size since the longest queue must be at least $\frac{B}{N}$ size when the buffer is full.
In essence, \name circumvents the impact of large prediction error by accepting packets until a certain amount of buffer is filled up.

\myitem{Predictions:}
\name can be used with any ML oracle that predicts whether to accept or drop a packet, according to our prediction model (see \S\ref{sec:prediction-model}). We do not rely on the internal details of the oracle. However, certain choices of ML oracles are better suited to operate within the limited resources available in a switch hardware. We discuss further on our choice of oracle later in \S\ref{sec:practicality}.

\begin{algorithm}[!t]
	\SetKwFunction{arrival}{\textbf{\textsc{\textcolor{myred}{arrival}}}}
	\SetKwFunction{updateThreshold}{\textbf{\textsc{\textcolor{myred}{updateThreshold}}}}
	\SetKwFunction{departure}{\textbf{\textsc{\textcolor{myred}{departure}}}}

	\SetKwProg{Fn}{function}{:}{}
	\SetKwProg{Proc}{procedure}{:}{}
	\SetKwInOut{KwIn}{Input}
	\SetKwInOut{KwOut}{Output}

	\KwIn{\ Packet arrivals $\sigma$, \\ \ \ Drop predictions $\phi^\prime(\sigma)$}

	\Proc{\arrival{$\sigma(t)$}}{

		\For{each packet $p \in \sigma(t)$}{
			Let $i$ be the destination queue for the packet $p$

			\textsc{updateThreshold}$(i, arrival)$

			\hspace*{-\fboxsep}\colorbox{alg1}{\parbox{0.8\linewidth}{%
					\Comment{\textcolor{darkgray}{\textit{Guarantees $N$-competitiveness}}}

					Let $j$ be the longest queue

					\If{$q_j(t) < \frac{B}{N}$}{

						$q_i(t) \leftarrow q_i(t) + 1$ \Comment{Accept}

						Continue to next packet
					}

				}}

			\hspace*{-\fboxsep}\colorbox{alg2}{\parbox{0.8\linewidth}{%
					\Comment{\textcolor{darkgray}{\textit{Enables $1.707\ \eta$-competitiveness}}}

					\label{line:thresholds}
					\eIf{$q_i(t) < T_i(t)$}{
						\If{$Q(t) < B$}{

							$drop$ = \textsc{GetPrediction}()

							\eIf{drop}{

								\Comment{Drop}

							}{
								$q_i(t) \leftarrow q_i(t) + 1$ \Comment{Accept}
							}

						}
					}{
						\Comment{Drop}
					}

				}}
		}

	}

	\Proc{\departure{i}}{

		\If{$q_i(t) > 0$}{

			$q_i(t) \leftarrow q_i(t) - 1$ \Comment{Drain one packet}

		}
		\textsc{updateThreshold}$(i, departure)$
	}

	\Proc{\updateThreshold{i, event}}{

		\If{event = arrival}{
			\hspace*{-\fboxsep}\colorbox{alg3}{\parbox{0.8\linewidth}{%
					\Comment{\textcolor{darkgray}{\textit{Thresholds are treated as LQD queue lengths}}}

					\eIf(\Comment{Sum of thresholds}){$\Gamma(t) = B$}{

						Let $T_j(t)$ be the highest threshold

						$T_j(t) \leftarrow T_j(t) - 1$ \Comment{Decrease}

						$T_i(t) \leftarrow T_i(t) + 1$ \Comment{Increase}
					}{
						$T_i(t) \leftarrow T_i(t) + 1$ \Comment{Increase}

						$\Gamma(t) \leftarrow \Gamma(t) + 1$
					}
				}}
		}
		\If{event = departure}{
			\If{$T_i(t) > 0$}{
				$T_i(t) \leftarrow T_i(t) - 1$ \Comment{Decrease}

				$\Gamma(t) \leftarrow \Gamma(t) -1$
			}
		}
	}

	\caption{\name}
	\label{alg:follow-pred}
\end{algorithm}

\subsection{Properties of \name}
\label{sec:properties}

\name offers attractive theoretical guarantees in terms of competitive ratio. In this section, for simplicity, we refer an offline optimal algorithm as $OPT$.

Although we have so far discussed the prediction error more intuitively, it requires a quantitative measure in order to analyze the performance of an algorithm relying on predictions.
There are two important considerations in defining a suitable error function.
First, following the literature, an error function
must be independent of the state and actions of our algorithm, so that we can train a predictor without considering all possible states of the algorithm~\cite{10.1145/3409964.3461790}. Second, it is desirable that the performance of our algorithm can be related to the error function in an uncomplicated manner. Taking these into consideration, we define our error function in Definition~\ref{def:error-function}. Our definition captures the prediction error in terms of the performance of LQD (push-out) and the performance of an algorithm $FollowLQD$. Here, $FollowLQD$ (Algorithm~\ref{alg:follow} in Appendix~\ref{app:detalg}) is a deterministic drop-tail algorithm (without predictions) with thresholds similar to \name.

\begin{restatable}[Error function]{definition}{errorFunction}\label{def:error-function}
	Let LQD($\sigma$) and FollowLQD($\sigma$) denote the total number of packets transmitted by the online push-out algorithm LQD and the online drop-tail algorithm FollowLQD over the arrival sequence $\sigma$.
	Let $\phi$ denote the sequence indicating drop by LQD for each packet in the arrival sequence $\sigma$.  Let $\phi^\prime$ denote the sequence of drops predicted by the machine-learned oracle. Let $\phi_{TP}^\prime$, $\phi_{FP}^\prime$, $\phi_{TN}^\prime$, and $\phi_{FN}^\prime$ denote the sequence of true positive, false positive, true negative and false negative predictions for the arrival sequence $\sigma$.
	We define the error function $\eta(\phi,\phi^\prime)$ as follows:

	\begin{align}
		\eta(\phi,\phi^\prime) & = \frac{LQD(\sigma)}{\displaystyle FollowLQD\left(\sigma - \phi^\prime_{TP} - \phi^\prime_{FP}\right)} \\ \nonumber
	\end{align}

\end{restatable}

Using Definition~\ref{def:error-function}, we analyze the throughput of \name over an entire packet arrival sequence $\sigma$ based on the predictions $\phi^\prime$. In fact, our error function is upper bounded by an intuitive closed form expression, in terms of the number of true and false predictions, as follows, that can be easily computed:\footnote{We prove our upper bound in Theorem~\ref{th:error-upper} (Appendix~\ref{app:predictions}).}
\[
	\eta(\phi,\phi^\prime) \le \frac{\phi^\prime_{TN} + \phi^\prime_{FP}}{\displaystyle \phi^\prime_{TN} - \displaystyle\min\left((N-1)\cdot\phi^\prime_{FN}, \phi^\prime_{TN}\right)}
\]
The upper bound of our error function indicates intuitively that \first the error decreases as the total number of true negative predictions dominate the total false predictions, \second the error increases with each false positive prediction and \third the error increases with each false negative with a larger weight.
Lemma~\ref{lemma:alg-lqd} states the relation between the throughput of \name, throughput of $LQD$ and the prediction error.

\begin{restatable}{Lemma}{lemmaAlgLqd}\label{lemma:alg-lqd}
	The total number of packets transmitted by \name for an arrival sequence $\sigma$, a drop sequence $\phi$ by LQD and the predicted drop sequence $\phi^\prime$ is given by
	\begin{align}\label{eq:alg-lqd}
		\name(\sigma) \ge \frac{LQD(\sigma)}{\underbrace{\eta(\phi, \phi^\prime)}_{error}}
	\end{align}
\end{restatable}

Equation~\ref{eq:alg-lqd} shows that the throughput of \name reaches closer to (moves away from) $LQD$ as the prediction error becomes smaller (larger). We present a sketch of our proof here. Our full proof appears in Appendix~\ref{app:predictions}. We begin by analyzing the drops incurred by \name based on the drop criterion described in \S\ref{sec:algorithm}.
We argue that for every true positive and false positive predictions, there is at most one drop by \name. All other drops incurred by \name are due to the thresholds. Using these observations, we show that \name transmits at least the number of packets transmitted by $FollowLQD$ over the arrival sequence $\sigma - \phi^\prime_{TP} - \phi^\prime_{FP}$. This leads us to Equation~\ref{eq:alg-lqd}.

Recall that \name bypasses the drop criterion and accepts packets through a safeguard condition under certain cases (see \S\ref{sec:algorithm}). Based on this, we obtain another bound for the throughput of \name in Lemma~\ref{lemma:alg-opt}, which is independent of the prediction error.

\begin{restatable}{Lemma}{lemmaAlgOpt}\label{lemma:alg-opt}
	\name transmits at least $\frac{1}{N}$ times the number of packets transmitted by an offline optimal algorithm OPT \ie $\name(\sigma) \ge \frac{1}{N} \cdot OPT(\sigma)$.
\end{restatable}

Lemma~\ref{lemma:alg-opt} shows that irrespective of the prediction error (even under large error), \name can always transmit at least $\frac{1}{N}$ fraction of the packets transmitted by an optimal solution. Our proof of Lemma~\ref{lemma:alg-opt} is based on the fact that upon a drop by \name, there is at least one queue with length $\frac{B}{N}$ (the safeguard condition). As a result, for every $B$ packets transmitted by $OPT$, there are at least $\frac{B}{N}$ number of packets transmitted by \name over the arrival sequence $\sigma$. This leads us to the bound expressed in Lemma~\ref{lemma:alg-opt}.

Finally, using the above two results, we prove the competitive ratio of \name as a function of the prediction error. \name's competitive ratio satisfies the three desirable properties: $2$-consistent, $N$-robust and exhibits smoothness.
\begin{restatable}{theorem}{theoremCRatioPred}\label{theorem:c-ratio-pred}
	The competitive ratio of \name grows linearly from $1.707$ to $N$ based on the prediction error $\eta(\phi,\phi^\prime)$, where $N$ is the number of ports, $\phi$ is the drop sequence of LQD and $\phi^\prime$ is the predicted sequence of drops \ie the competitive ratio is at most $\min(1.707\ \eta(\phi,\phi^\prime), N)$.
\end{restatable}

Our proof follows from Lemmas~\ref{lemma:alg-lqd}~and~\ref{lemma:alg-opt} (see Appendix~\ref{app:predictions}). Theorem~\ref{theorem:c-ratio-pred} essentially shows how \name's competitive ratio in terms of throughput improves from $N$ to $1.707$ as the prediction error (Definition~\ref{def:error-function}) decreases.
We note that our analysis compares an algorithm against an optimal offline algorithm over a fixed packet arrival sequence. This allows us to analyze the competitive ratio via an error function defined over the corresponding arrival sequence. However, real-world traffic is responsive in nature due to congestion control and packet retransmissions. Although we have used $\eta$ as our error function to express the competitiveness of \name, in our evaluation (\S\ref{sec:evaluation}), we compare \name with state-of-the-art approaches under realistic datacenter workloads and we also present the quality of our predictions using more natural error functions that are widely used for machine learning models.

\subsection{Practicality of \name}
\label{sec:practicality}

\name's algorithm itself is simple and close to complexity of the longest queue drop (push-out). However, the machine-learned oracle producing the predictions adds additional complexity in order to deploy \name on switches. Overall, there are three main parts of \name that contribute to additional complexity in terms of memory and computation: \first finding the longest queue (and its threshold), \second remembering thresholds and \third obtaining predictions.

\myitem{Finding the longest queue (and its threshold):} For every packet arrival, \name requires finding the longest queue for the safeguard condition described in \S\ref{sec:algorithm}. Additionally, \name requires finding the largest threshold during the threshold updates upon every packet arrival. The maximum value search operation has a run-time complexity of $\mathcal{O}(N)$, where $N$ is the number of ports. Note that typical datacenter switches have a relatively small number of ports \eg $64$ ports in Broadcom Tomahawk4~\cite{broadcom}. Prior work in the context of LQD proposes an approximation to further reduce the complexity of finding the longest queue~\cite{tcpgigabit} to $\mathcal{O}(1)$.
The average case complexity can further be reduced by only maintaining the list of queue lengths (and their thresholds) that are larger than $\frac{B}{N}$. This is sufficient since the safeguard condition checks whether the longest queue is less than $\frac{B}{N}$, which is the same as checking that no queue is longer than $\frac{B}{N}$. Similarly, the largest threshold search during the threshold updates is only triggered when the buffer is full. In this case, the longest queue must be at least $\frac{B}{N}$.
Given that switches are becoming more and more computationally capable, we believe that a basic function such as finding the maximum value in a small list is feasible to implement within the available~resources.

\begin{figure*}
	\centering
	\begin{subfigure}{1\linewidth}
		\centering
		\includegraphics[width=0.8\linewidth]{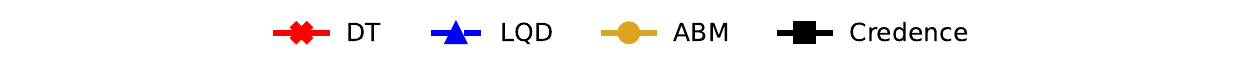}
		\vspace{-3mm}
	\end{subfigure}
	\begin{subfigure}{0.248\linewidth}
		\centering
		\includegraphics[width=1\linewidth]{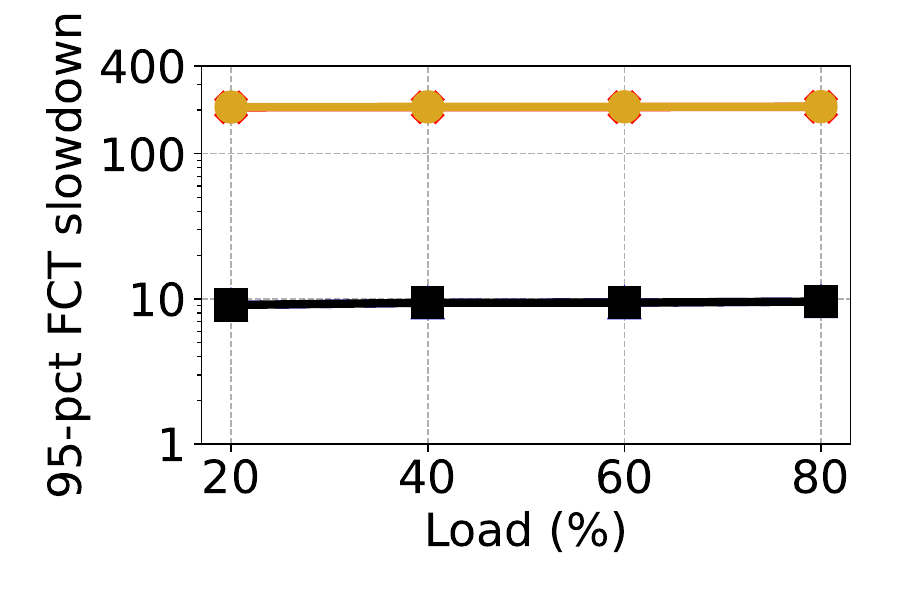}
		\caption{Incast flows}
		\label{fig:dctcp-loads-incastfct}
	\end{subfigure}\hfill
	\begin{subfigure}{0.248\linewidth}
		\centering
		\includegraphics[width=1\linewidth]{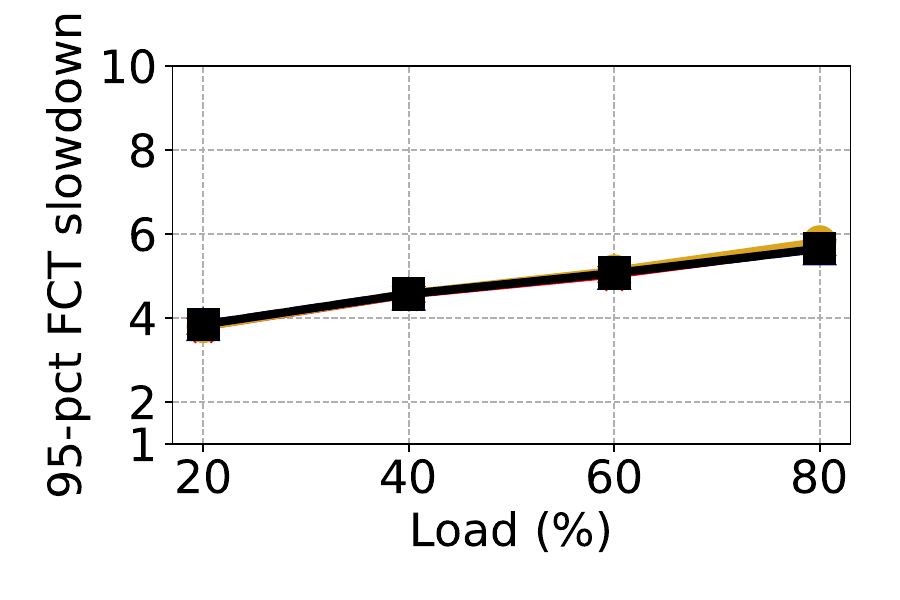}
		\caption{Short flows}
		\label{fig:dctcp-loads-shortfct}
	\end{subfigure}\hfill
	\begin{subfigure}{0.248\linewidth}
		\centering
		\includegraphics[width=1\linewidth]{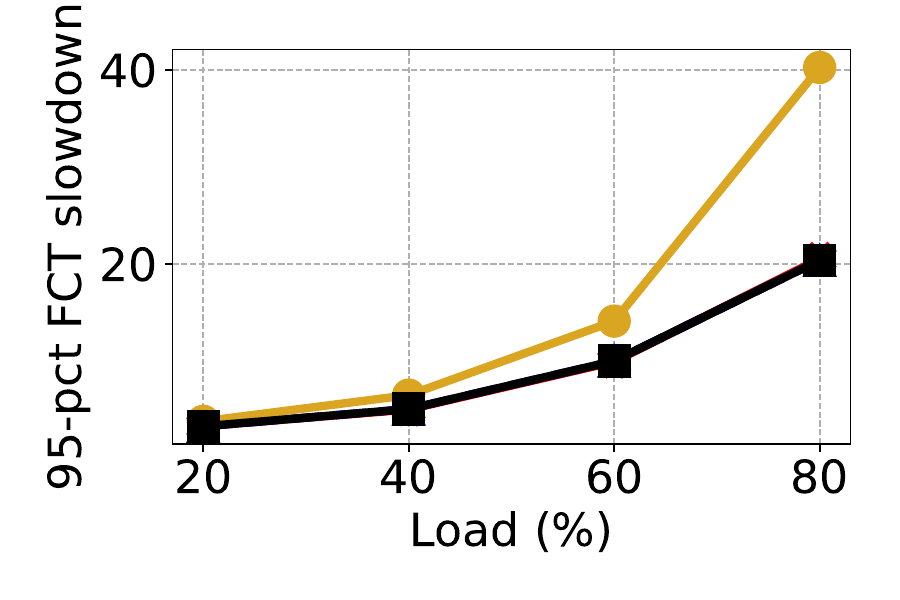}
		\caption{Long flows}
		\label{fig:dctcp-loads-longfct}
	\end{subfigure}\hfill
	\begin{subfigure}{0.248\linewidth}
		\centering
		\includegraphics[width=1\linewidth]{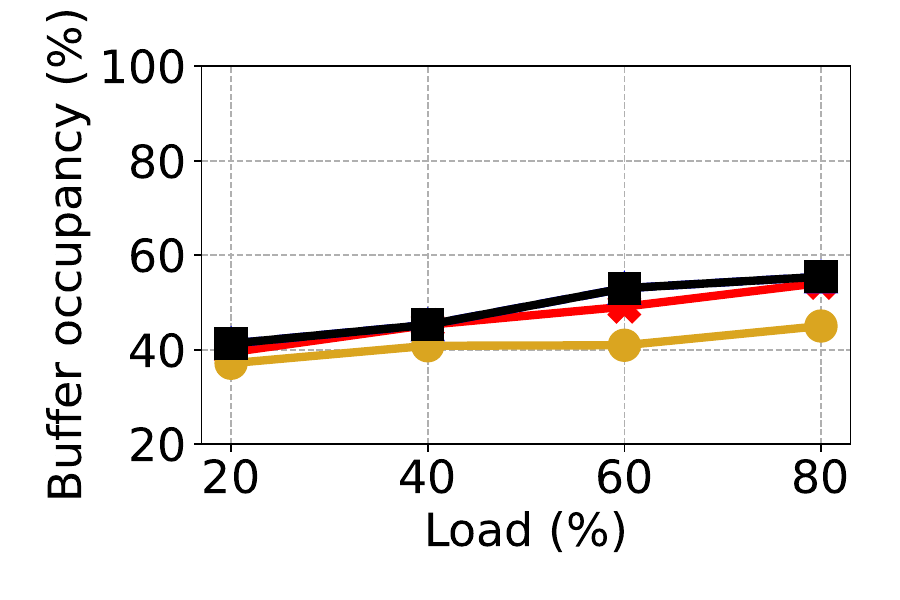}
		\caption{Shared buffer occupancy}
		\label{fig:dctcp-loads-buffer}
	\end{subfigure}\hfill
	\vspace{-2mm}
	\caption{Performance of \name across various loads of websearch workload and incast workload at a burst size $50$\% of the buffer size, with DCTCP as the transport protocol. As the load increases, ABM penalizes long flows. DT and ABM are unable to absorb bursts of size $50$\% of the buffer size. \name achieves superior burst absorption and does not penalize long flows.}
	\label{fig:dctcp-loads}
	\vspace{-3mm}
\end{figure*}
\begin{figure*}
	\centering
	\begin{subfigure}{0.248\linewidth}
		\centering
		\includegraphics[width=1\linewidth]{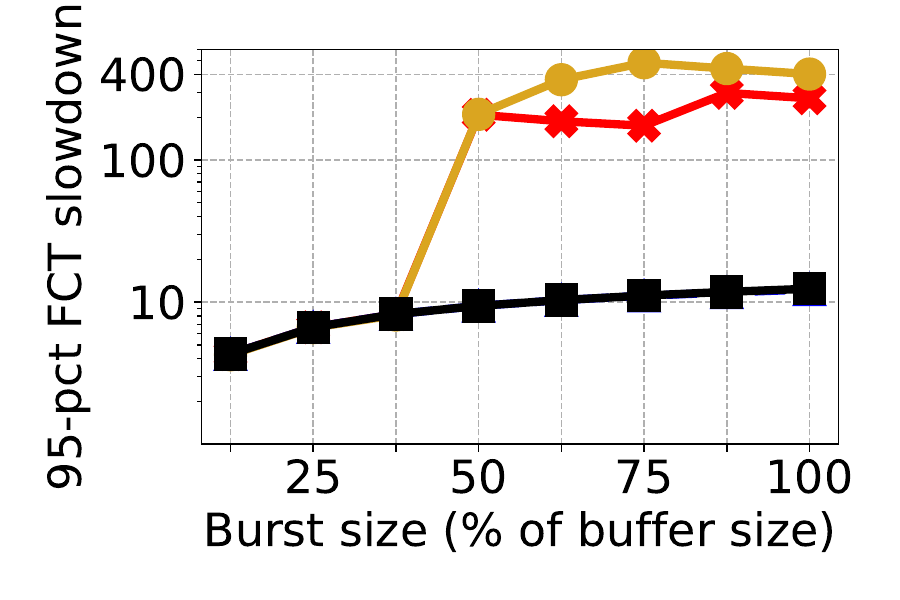}
		\caption{Incast flows}
		\label{fig:dctcp-bursts-incastfct}
	\end{subfigure}\hfill
	\begin{subfigure}{0.248\linewidth}
		\centering
		\includegraphics[width=1\linewidth]{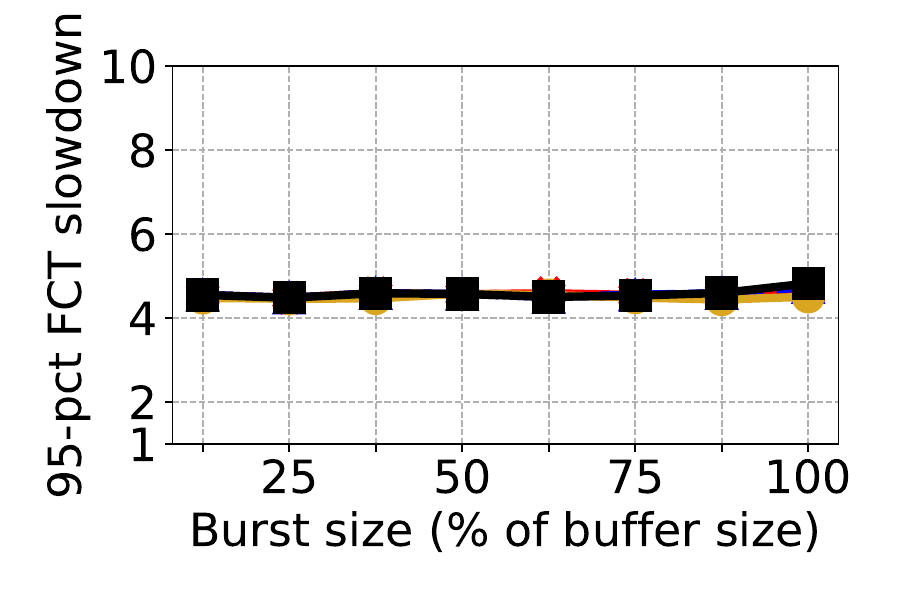}
		\caption{Short flows}
		\label{fig:dctcp-bursts-shortfct}
	\end{subfigure}\hfill
	\begin{subfigure}{0.248\linewidth}
		\centering
		\includegraphics[width=1\linewidth]{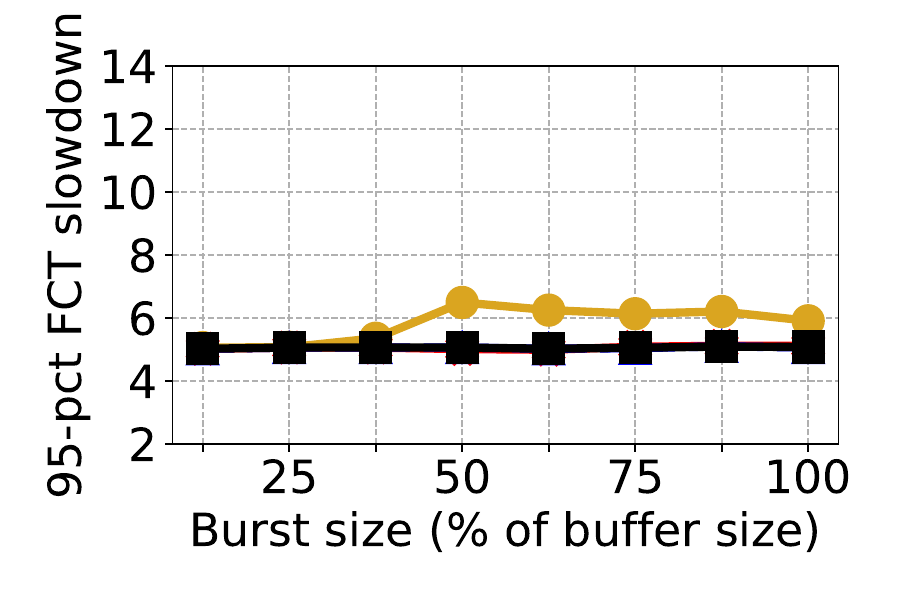}
		\caption{Long flows}
		\label{fig:dctcp-bursts-longfct}
	\end{subfigure}\hfill
	\begin{subfigure}{0.248\linewidth}
		\centering
		\includegraphics[width=1\linewidth]{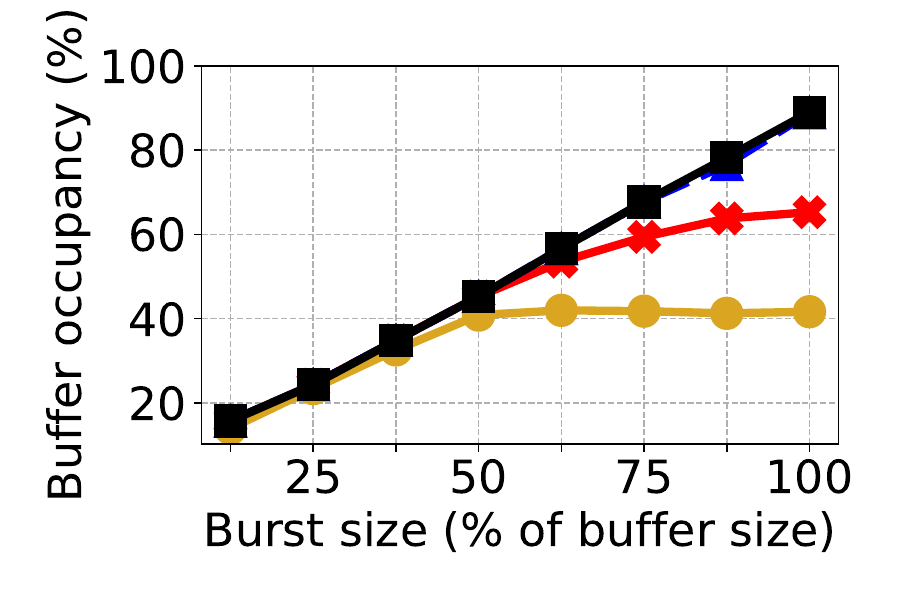}
		\caption{Shared buffer occupancy}
		\label{fig:dctcp-bursts-buffer}
	\end{subfigure}\hfill
	\vspace{-2mm}
	\caption{Performance of \name across various burst sizes of incast workload and websearch workload at $40$\% load, with DCTCP as the transport protocol. At small burst sizes, DT and ABM achieve similar performance compared to \name but as the burst size increases, \name outperforms DT and ABM in terms of FCTs for incast flows (burst absorption).}
	\label{fig:dctcp-bursts}
	\vspace{-3mm}
\end{figure*}

\myitem{Thresholds memory:} In contrast to existing threshold-based algorithms, \name's thresholds depend on their previous value \ie thresholds must be remembered. As a result, \name adds a small memory overhead of $\mathcal{O}(N)$ for the thresholds. The threshold calculations are in fact much simpler than existing schemes and do not add any further computational complexity since \name only requires adding and subtracting the threshold values by the packet size.

\myitem{Predictions:} Our prediction model (drop or accept) essentially boils down to binary classification problem. To this end, numerous ML techniques exist ranging from linear classifiers to more advanced neural networks. In view of practicality, we consider random forests as they are implementable in programmable hardware~\cite{10229100,busse2019pforest}. In order to reduce the prediction latency, we also limit the number of trees and the maximum depth of our trained random forest model. We find that, even a model trained with a maximum depth of four, and as low as four to eight trees achieves reasonable prediction error (precision $\approx 0.65$). Further, to reduce the complexity of the model, we also limit the number of features to four: queue length, total shared buffer occupancy and their corresponding moving averages (exponentially weighted) over one round-trip time.

The fundamental blocks required for \name are all individually practical in today's hardware. Unfortunately, modifying the buffer sharing algorithm and integrating it with predictions requires switch vendor support. Even in programmable switches, the traffic manager is merely a blackbox that implements Dynamic Thresholds with a single parameter exposed to the user. Given the superior performance of \name (\S\ref{sec:evaluation}), we wish to gain attention from switch vendors to discuss further on the implementation of \name.

\section{Evaluation}
\label{sec:evaluation}
We evaluate the performance of \name and compare it against state-of-the-art buffer sharing algorithms in the context of datacenter networks. Our evaluation aims at answering three main questions:

\myitemit{\textbf{(Q1)} Does} \name \textit{improve the burst absorption?}

\noindent Our evaluation shows that \name significantly improves the burst absorption capabilities of switches. We find that \name improves the $95$-percentile flow completion times for incast flows by up to $95.4$\% compared to Dynamic Thresholds (DT) and by up to $96.9$\% compared to ABM.

\myitemit{\textbf{(Q2)} Can} \name \textit{improve the flow completion times for short flows as well as long flows?}

\noindent We find that \name performs similar to existing approaches in terms of $95$-percentile flow completion times for short flows and improves upon ABM by up to $22$\% correspondingly for long flows.

\myitemit{\textbf{Q3} How does prediction error impact the performance of} \name\textit{ in terms of flow completion times?}

\noindent We increase the error of our prediction by artificially flipping the predictions with a probability. As the probability increases (error increases), we find that \name sustains performance up to $0.01$ probability and smoothly degrades in performance beyond $0.01$.

\subsection{Setup}
Our evaluation is based on packet-level simulations in NS3~\cite{ns3}. We embed a Python interpreter within NS3 using pybind11~\cite{pybind11} in order to obtain predictions from a random forest model trained with scikit-learn~\cite{scikit}.

\myitem{Topology:} We consider a leaf-spine topology with $256$ servers organized into $4$ spines and $16$ leaves. Each link has a propagation delay of $3\mu$s leading to a round-trip-time of $25.2\mu$s. The capacity is set to $10$Gbps for all the links leading to $4:1$ oversubscription similar to prior works~\cite{278346,10.1145/3387514.3405899,abm}. All the switches in our topology have $5.12$KB buffer-per-port-per-Gbps similar to Broadcom Tomahawk~\cite{broadcom}.

\myitem{Workloads:} We generate traffic using websearch~\cite{10.1145/1851182.1851192} flow size distribution that is based on measurements from real-world datacenter workloads. We vary the load on the network in the range $20$-$80$\%. We additionally generate traffic using a synthetic incast workload similar to prior work~\cite{abm}. Our incast workload mimics the query-response behavior of a distributed file storage system where each query results in a bursty response from multiple servers. We set the query request rate to $2$ per second from each server, and we vary the burst size in the range $10$-$100$\% of the switch buffer size. We use DCTCP~\cite{10.1145/1851182.1851192} and PowerTCP~\cite{278346} as transport protocols.

\begin{figure*}
	\centering
	\begin{subfigure}{1\linewidth}
		\centering
		\includegraphics[width=0.8\linewidth]{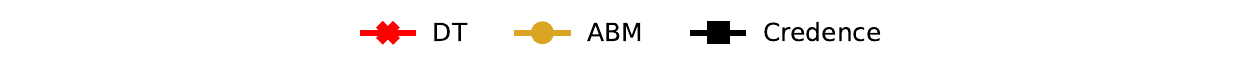}
		\vspace{-3mm}
	\end{subfigure}
	\begin{subfigure}{0.248\linewidth}
		\centering
		\includegraphics[width=1\linewidth]{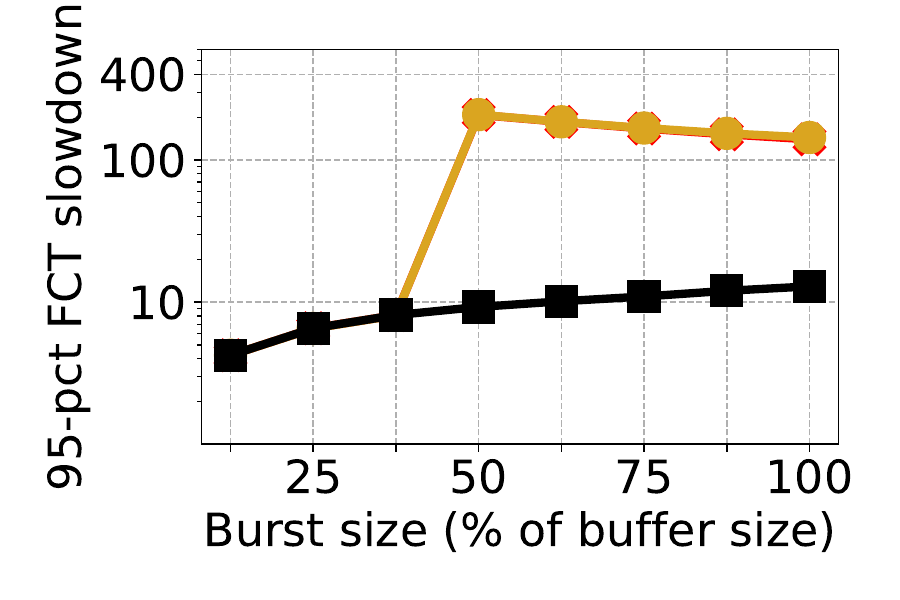}
		\caption{Incast flows}
		\label{fig:powertcp-bursts-incastfct}
	\end{subfigure}\hfill
	\begin{subfigure}{0.248\linewidth}
		\centering
		\includegraphics[width=1\linewidth]{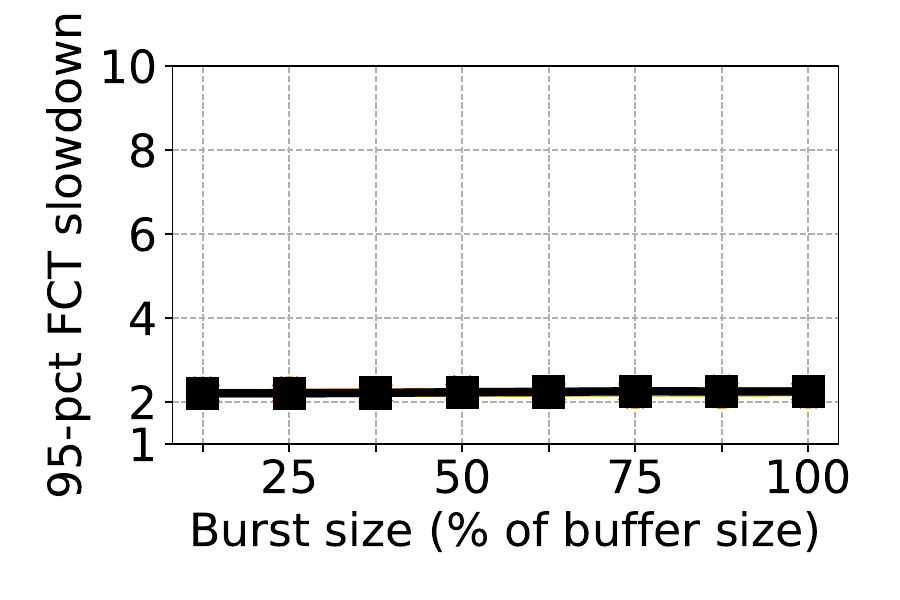}
		\caption{Short flows}
		\label{fig:powertcp-bursts-shortfct}
	\end{subfigure}\hfill
	\begin{subfigure}{0.248\linewidth}
		\centering
		\includegraphics[width=1\linewidth]{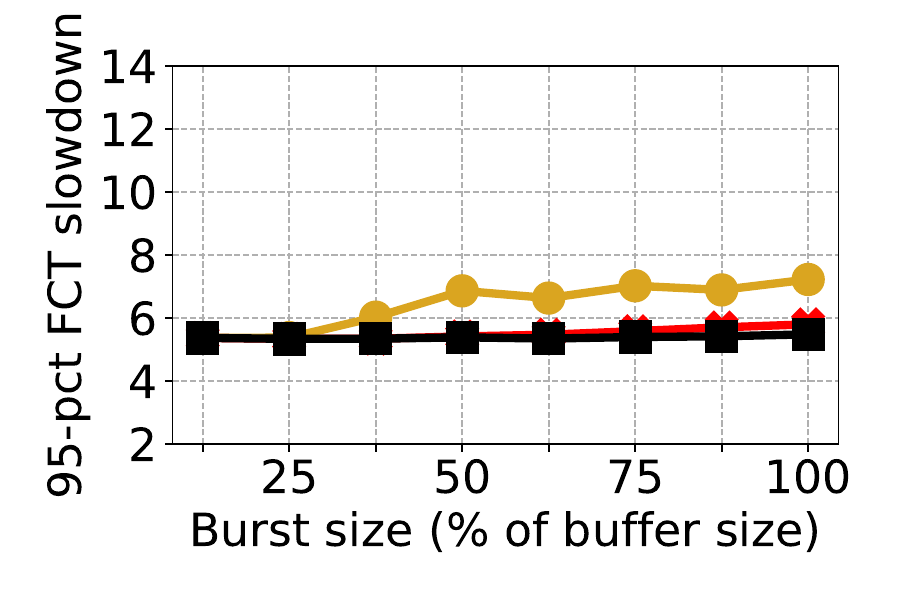}
		\caption{Long flows}
		\label{fig:powertcp-bursts-longfct}
	\end{subfigure}\hfill
	\begin{subfigure}{0.248\linewidth}
		\centering
		\includegraphics[width=1\linewidth]{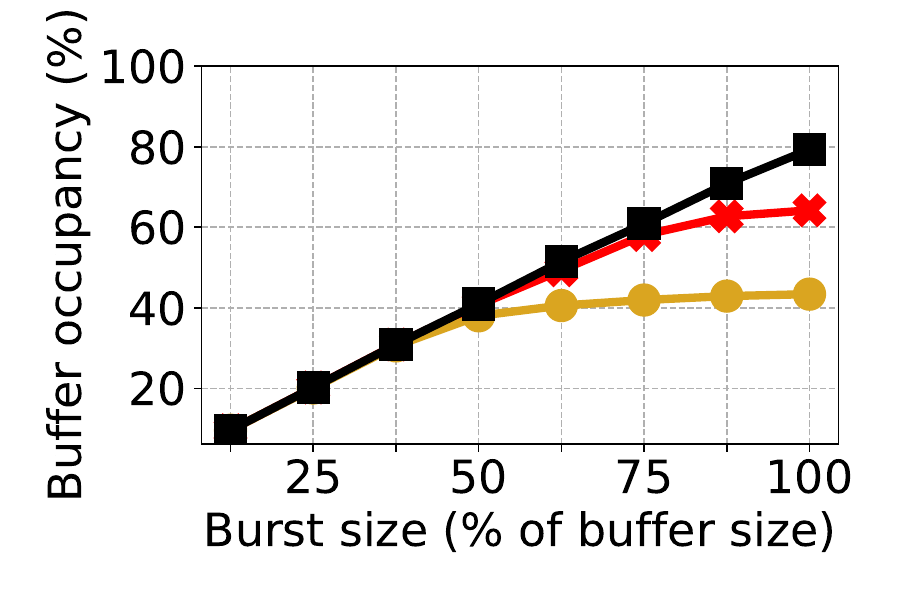}
		\caption{Shared buffer occupancy}
		\label{fig:powertcp-bursts-buffer}
	\end{subfigure}\hfill
	\vspace{-2mm}
	\caption{Performance of \name across various burst sizes of incast workload and websearch workload at $40$\% load, with PowerTCP as the transport protocol. Even with advanced congestion control, DT and ABM only benefit in terms of FCTs for long flows, but \name outperforms in terms of FCTs for incast flows (burst absorption) as well as FCTs for long flows.}
	\label{fig:powertcp-bursts}
\end{figure*}

\myitem{Comparisons \& metrics:} We compare \name with Dynamic Thresholds~\cite{choudhury1998dynamic} (the default algorithm in datacenter switches), ABM~\cite{abm} and LQD (push-out). Hereafter, we refer to Dynamic Thresholds as DT. We report four performance metrics: $95$-percentile flow completion times for short flows ($\le 100$KB), incast flows (incast workload), long flows ($\ge 1$MB), and the $99$-percentile shared buffer occupancy. We present the CDFs of flow completion times in~Appendix~\ref{app:results}.

\noindent\textbf{Predictions:} We collect packet-level traces from each switch in our topology when using LQD (push-out) as the buffer sharing algorithm. Each trace consists of five values corresponding to each packet arrival: \first queue length, \second average queue length, \third shared buffer occupancy, \fourth average shared buffer occupancy and \fifth accept (or drop). We train a random forest classifier using queue length and shared buffer occupancy as features and the model predicts packet drops. We set the maximum depth for each tree to $4$ in view of practicality. At a train-test split $0.6$ of our LQD trace, based on our parameter sweep across the number of trees used for our classifier (Figure~\ref{fig:scores} in Appendix~\ref{app:results}), we set the number of trees to $4$. We observe that the quality of our predictions does not improve significantly beyond four trees in our datasets. This gives us an accuracy of $0.99$, error score $\frac{1}{\eta}$ $0.99$\footnote{The high values of accuracy and our error score $\frac{1}{\eta}$ are attributed to the dataset being skewed \ie congestion is not persistent.} (inverse of our error function based on Definition~\ref{def:error-function}), precision of $0.65$, recall of $0.35$ and F1 score of $0.45$. We defer the definitions of these prediction scores to Appendix~\ref{app:predictions} as they are standard in the literature. We train our model with an LQD trace corresponding to websearch workload at $80$\% load, and a burst size of $75$\% buffer size for the incast workload, using DCTCP as the transport protocol.
We use the same trained model in all our evaluations. We ensure that our test scenarios are different from the training dataset by using different random seeds in addition to different traffic conditions (different loads and different burst sizes) in each~experiment in our evaluation.

\myitem{Configuration:} \name is parameter-less, and it takes input from an oracle (described above) that predicts packet drops. We set $\alpha=0.5$ for DT and ABM similar to prior work~\cite{abm}. ABM uses $\alpha=64$ for all the packets which arrive during the first round-trip-time~\cite{abm}. We configure DCTCP according to~\cite{10.1145/1851182.1851192} and PowerTCP according to~\cite{278346}.

\subsection{Results}

\begin{figure*}
	\centering
	\begin{subfigure}{1\linewidth}
		\centering
		\includegraphics[width=0.8\linewidth]{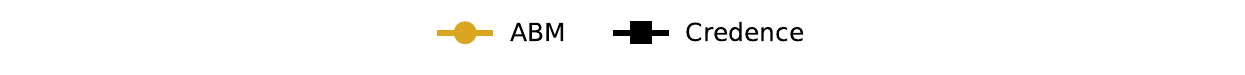}
		\vspace{-3mm}
	\end{subfigure}
	\begin{subfigure}{0.248\linewidth}
		\centering
		\includegraphics[width=1\linewidth]{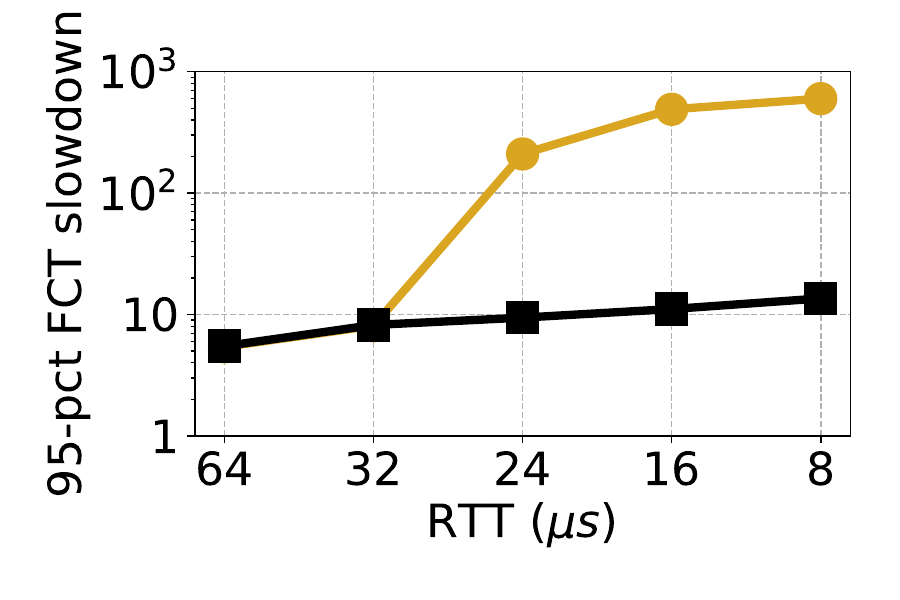}
		\caption{Incast flows}
		\label{fig:latency-incastfct}
	\end{subfigure}\hfill
	\begin{subfigure}{0.248\linewidth}
		\centering
		\includegraphics[width=1\linewidth]{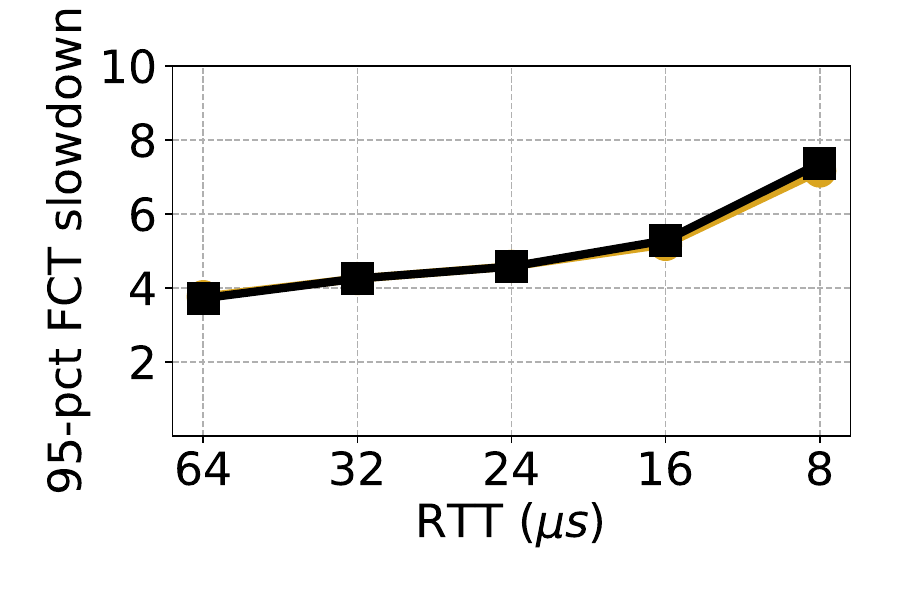}
		\caption{Short flows}
		\label{fig:latency-shortfct}
	\end{subfigure}\hfill
	\begin{subfigure}{0.248\linewidth}
		\centering
		\includegraphics[width=1\linewidth]{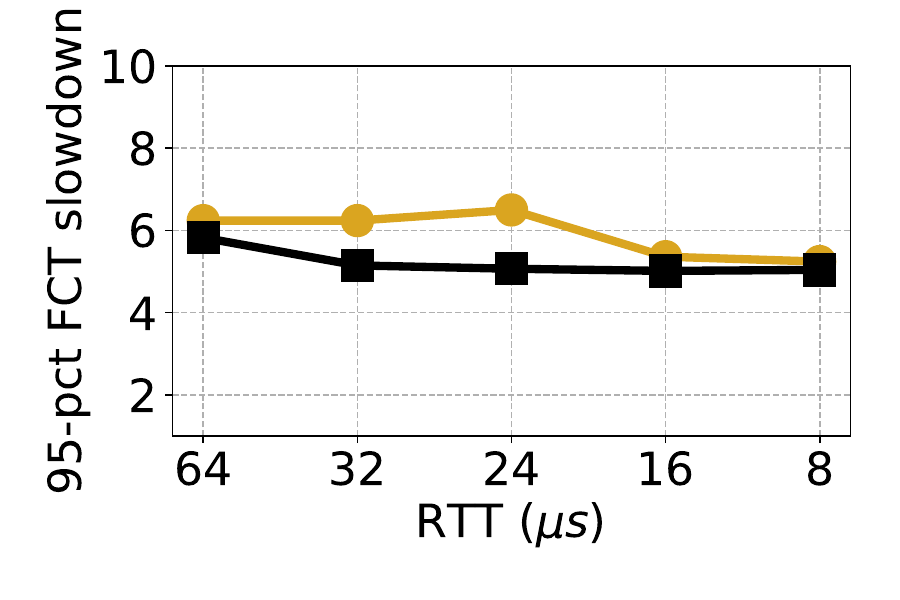}
		\caption{Long flows}
		\label{fig:latency-longfct}
	\end{subfigure}\hfill
	\begin{subfigure}{0.248\linewidth}
		\centering
		\includegraphics[width=1\linewidth]{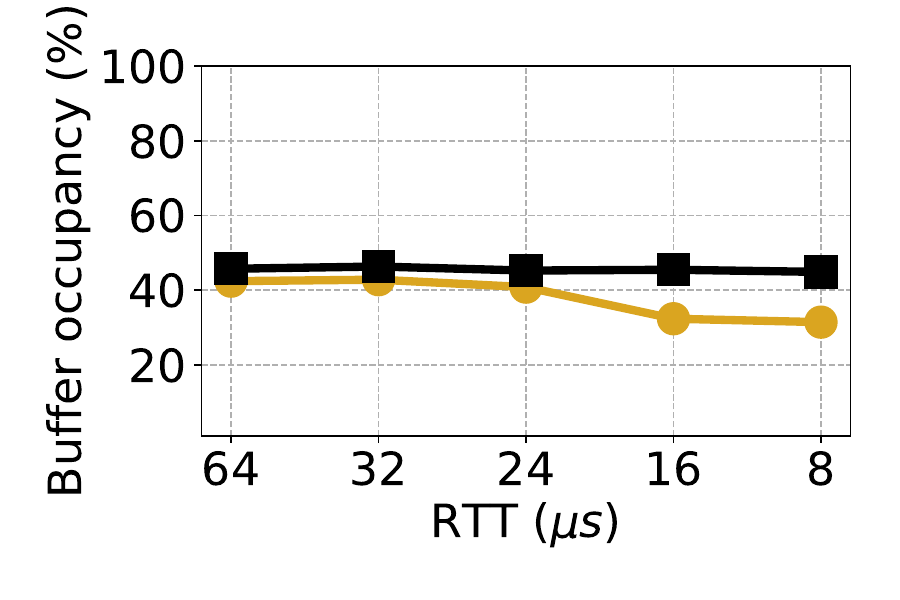}
		\caption{Shared buffer occupancy}
		\label{fig:latency-buffer}
	\end{subfigure}\hfill
	\vspace{-2mm}
	\caption{ABM is sensitive to RTT and performs significantly worse compared to \name at low RTTs. At high RTTs, ABM performs similar to \name.}
	\label{fig:latency}
	\vspace{-5mm}
\end{figure*}
\begin{figure*}
	\centering
	\begin{subfigure}{1\linewidth}
		\centering
		\includegraphics[width=0.8\linewidth]{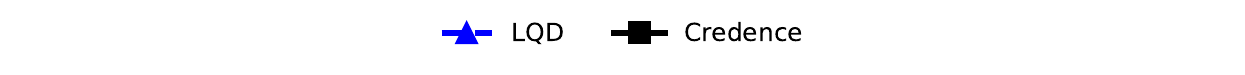}
		\vspace{-3mm}
	\end{subfigure}
	\begin{subfigure}{0.248\linewidth}
		\centering
		\includegraphics[width=1\linewidth]{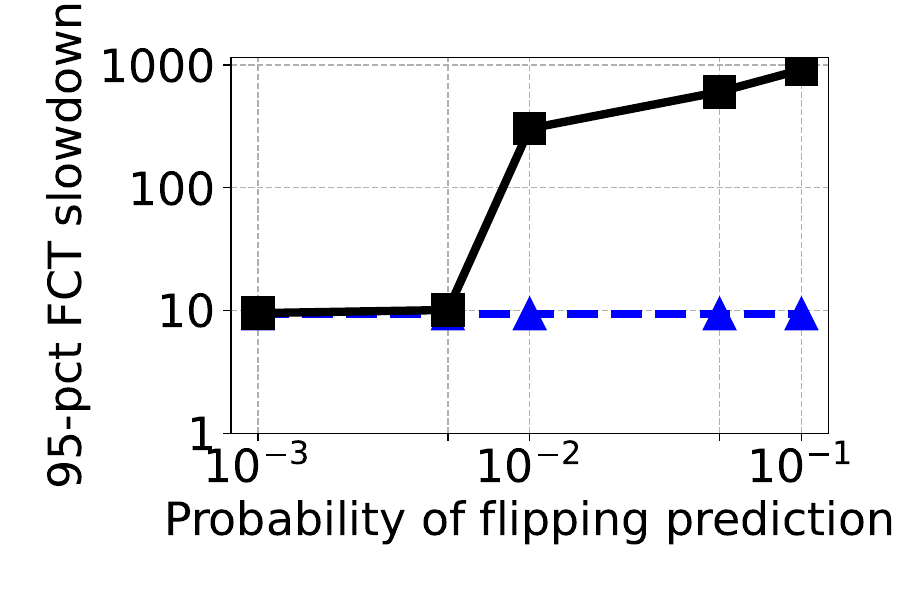}
		\caption{Incast flows}
		\label{fig:error-incastfct}
	\end{subfigure}\hfill
	\begin{subfigure}{0.248\linewidth}
		\centering
		\includegraphics[width=1\linewidth]{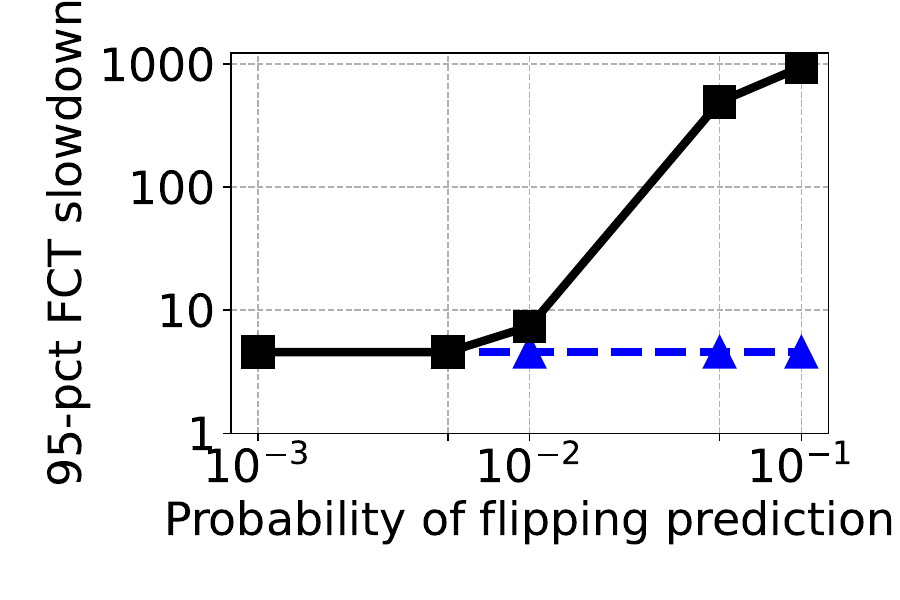}
		\caption{Short flows}
		\label{fig:error-shortfct}
	\end{subfigure}\hfill
	\begin{subfigure}{0.248\linewidth}
		\centering
		\includegraphics[width=1\linewidth]{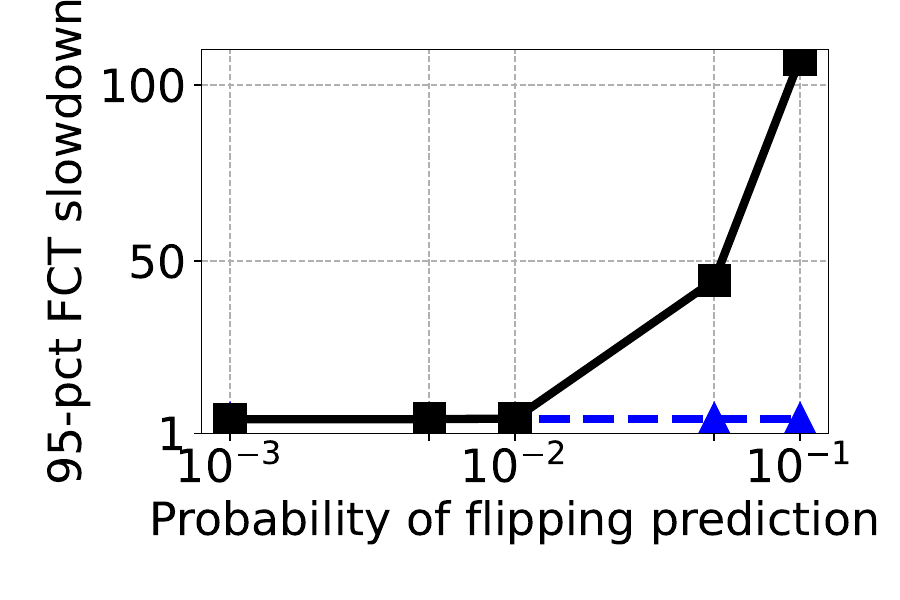}
		\caption{Long flows}
		\label{fig:error-longfct}
	\end{subfigure}\hfill
	\begin{subfigure}{0.248\linewidth}
		\centering
		\includegraphics[width=1\linewidth]{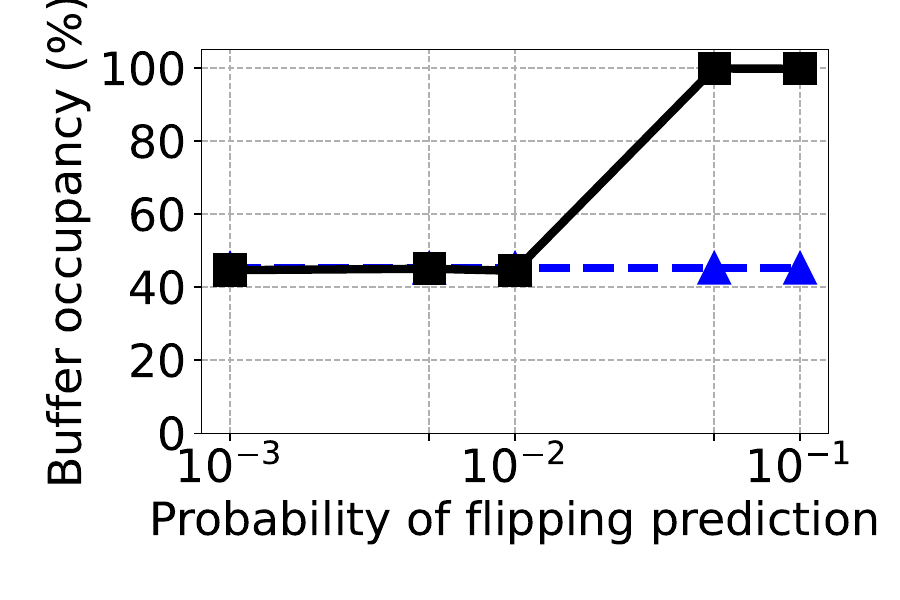}
		\caption{Shared buffer occupancy}
		\label{fig:error-buffer}
	\end{subfigure}\hfill
	\vspace{-2mm}
	\caption{Even though the predictions from our random forest classifier are intentionally flipped (to increase error), \name performs close to LQD up to $0.005$ flipping probability but smoothly diverges from LQD at $0.01$ flipping probability.}
	\label{fig:error}
	\vspace{-1mm}
\end{figure*}

\myitem{\name significantly improves burst absorption:}
In Figure~\ref{fig:dctcp-loads-incastfct}, using DCTCP as the transport protocol, we generate websearch traffic across various loads in the range $20$-$80$\% and generate incast traffic with a burst size of $50$\% buffer size.  We observe that \name performs close to the optimal performance of LQD. \name improves the $95$-percentile flow complete times for incast flows by $95.50$\% compared to DT, and by $95.53$\% compared to ABM. In Figure~\ref{fig:dctcp-bursts-incastfct}, we set the load of websearch traffic at $40$\% and vary the burst size for incast workload in the range $10$-$100$\% buffer size. \name performs similar to DT and ABM for small burst sizes. As the burst size increases, \name improves the $95$-percentile flow completion times for incast workload by $95$\% on average compared to DT, and by $96.92$\% on average compared to ABM. In Figure~\ref{fig:powertcp-bursts-incastfct}, even when using PowerTCP as the transport protocol, we see that \name improves the $95$-percentile flow completion times for incast flow by $93.27$\% on average compared to DT and by $93.36$\% on average compared to ABM. In terms of burst absorption, both DT and ABM are drop-tail algorithms, hence they face the drawbacks discussed in \S\ref{sub:drawbacks}. \name relies on predictions and unlocks the performance of LQD (push-out) as shown by our results in Figure~\ref{fig:dctcp-loads-incastfct} and Figure~\ref{fig:dctcp-bursts-incastfct}.

\myitem{\name improves long flows FCTs:} \name not only improves the burst absorption but also improves the flow completion times for long flows. In Figure~\ref{fig:dctcp-loads-longfct}, at $50$\% burst size for incast workload and across various loads of websearch workload, we observe that \name performs similar to DT in terms of $95$-percentile flow completion times for long flows and improves upon ABM on average by $28.49$\%. At $80$\% load, \name improves upon ABM by $49.34$\%. Across various burst sizes, and at $40$\% load of websearch workload, we observe from Figure~\ref{fig:dctcp-bursts-longfct} that \name improves upon ABM by up to $22.02$\% and by $12.02$\% on average. With PowerTCP as the transport protocol (Figure~\ref{fig:powertcp-bursts-longfct}), \name improves the $95$-percentile flow completion times for long flows by $3.31$\% on average compared to DT and by $17.35$\% compared to ABM. At a burst size of $100$\% buffer size, \name improves the flow completion times by $5.49$\% compared to DT and by $24.09$\% compared to ABM. As described in \S\ref{sub:drawbacks}, drop-tail algorithms such as DT and ABM cannot effectively navigate proactive and reactive drops, resulting in throughput loss \ie high flow completion times for long flows. In contrast, predictions guide \name to effectively navigate proactive and reactive drops. This allows \name to achieve better flow completion times even for long flows.

\myitem{\name does not waste buffer resources:} In anticipation of future burst arrivals, both DT and ABM buffer resources. We show the $99.99$-percentile buffer occupancies\footnote{DT, ABM and \name have similar tail occupancies ($100$-percentile) that occurs at rare congestion events in our simulations.} in Figure~\ref{fig:dctcp-loads-buffer} for various loads of websearch workload and at a burst size of $50$\% buffer size for incast workload. We observe that, DT (ABM) utilizes $3.77$\% ($18.68$\%) lower buffer space on average compared to \name, at the cost of increased flow completion times even for long flows. Even as the burst size increases (Figure~\ref{fig:dctcp-bursts-buffer}), DT and ABM are unable to efficiently utilize the buffer space. In contrast, \name efficiently utilizes the available buffer space as the burst size increases, improving burst absorption without sacrificing flow completion times for long flows.

\myitem{ABM is sensitive to RTT:}
Although ABM is expected to outperform DT, our evaluation results especially in terms of flow completion times for incast flows contradict the results presented in~\cite{abm}. We ran several simulations varying all the parameters in our setup in order to better understand the performance of ABM. We found that ABM is in fact sensitive to round-trip-time (RTT). We vary the base RTT of our topology in Figure~\ref{fig:latency} and compare \name with ABM. At high RTTs, we observe that ABM performs close to \name, but degrades in performance as RTT decreases. Specifically, at $8\mu$s RTT, ABM performs $97.73$\% worse compared to \name in terms of flow completion times for incast flows. Although ABM achieves on-par flow completion times for short flows, we observe that ABM degrades in flow completion times for long times as well as under-utilizes the buffer as RTT decreases. The poor performance of ABM at low RTTs is due to the fact that ABM prioritizes the first RTT packets and considers the rest of the traffic as steady-state traffic. However, it is not uncommon that datacenter switches experience bursts for several RTTs~\cite{burstimc2022}. Further, congestion control algorithms require multiples RTTs to converge to steady-state. In contrast, \name is parameter-less and does not make such assumptions. \name performs significantly better than existing approaches even with an off-the-shelf machine-learned predictor with a simple model.

\myitem{\name gradually degrades with prediction error:} Our random forest classifier that we used in our evaluations so far, has a precision close to $0.65$. In order to evaluate the performance of \name with even worse prediction error, we artificially introduce error by flipping every prediction obtained from our random forest classifier with a certain probability. We consider LQD (push-out) as a baseline since \name is expected to perform close to LQD and degrade as the prediction error grows large. Figure~\ref{fig:error} presents our evaluation results, under websearch workload at $40$\% and burst size $50$\% of the buffer size for incast workload. At $0.001$ flipping probability, \name performs close to LQD. However, at $0.01$ flipping probability \name starts to diverge from LQD and gets significantly worse at $0.1$ flipping probability. Figure~\ref{fig:error} gives practical insights into smoothness of \name in addition to our analysis.

\section{Related Work}
The buffer sharing problem has been widely studied for many decades. Research works in the literature range from push-out as well as drop-tail algorithms tailored for ATM networks~\cite{choudhury1998dynamic,tcpgigabit,749262,368919,477712} to more recent drop-tail algorithms tailored for datacenter networks~\cite{fab,abm,trafficaware,7859368,reveriensdi23}.  While we focus on the buffer sharing problem in this paper, several related but orthogonal approaches also tackle buffer problems in datacenter networks \eg end-to-end congestion control~\cite{278346,10.1145/3341302.3342085,10.1145/3387514.3406591,276958,10.1145/3544216.3544235,10.1145/1851182.1851192}, AQM~\cite{codel,6602305,red}, scheduling~\cite{10.1145/2486001.2486031,10.1145/2377677.2377710}, packet deflection~\cite{10.1145/2592798.2592806} and load-balancing~\cite{10.1145/2619239.2626316,10.1145/2890955.2890968,10.1145/3098822.3098839}. These approaches aim at reducing congestion events and the overall buffer requirements, but they cannot fundamentally address buffer contention across multiple switch ports sharing the same buffer.
Research on algorithms with predictions for various problems has recently been an active field of research~\cite{im2023online,NEURIPS2018ML,10.1145/3528087,mitzenmacher2021queues,mitzenmacher2019scheduling,NEURIPS2018_0f49c89d} but ours is the first approach tackling the buffer sharing problem with predictions. Ongoing research efforts  show the feasibility of deploying machine-learned predictions in the network data plane~\cite{busse2019pforest,10.1145/3365609.3365864,10229100}.

\section{Future Research Directions}
\label{sec:future-work}
\name is the first approach showing the performance benefits and guarantees by augmenting buffer sharing algorithms with predictions. We believe that this paper opens several interesting avenues for future work both in systems and theory. In this section, we discuss some of the future work directions to push approaches such as \name to be deployed in the real-world (\S\ref{sec:systems-future}), as well as to improve the performance guarantees offered by such approaches (\S\ref{sec:theory-future}).

\subsection{Systems for In-Network Predictions}\label{sec:systems-future}
In this paper, we show how predictions can improve the performance of drop-tail buffer sharing. Many interesting systems research questions remain in order to integrate buffer sharing and predictions in the network data plane.

\myitem{Training the model:} Achieving consistent performance requires lowering the prediction error. The training phase of the model is a critical step towards reaching optimal buffer sharing.
Although, for simplicity, we have used only four features to train our random forest model, it would be interesting to study the tradeoff that may arise between prediction error and the complexity of the model both in space and time.
Further, the trained model must be simple enough that fits within the resources available in the data plane. Developing such trained models is an important step forward.

\myitem{Deploying the model:} Recent works propose practical implementations for in-network machine-learning models \eg in the context of traffic classification~\cite{busse2019pforest}. P4 implementation of a model that predicts drops would enhance not only the practicality but also stimulate further research to design algorithms with performance guarantees better than \name.

\myitem{Alternative predictions:} As mentioned in \S\ref{sub:hope}, there are several different prediction models that can be considered for the buffer sharing problem. For instance, instead of predicting the drops, an oracle could predict packet arrivals just for a tiny window of the near future. Systems research on studying the practicality and deployability of different prediction models is a valuable future direction that would better guide the design of algorithms with predictions for the buffer sharing problem.

\myitem{Understanding push-out complexity:}
While push-out algorithms raised much interest initially, over the last years, research on this approach has been less active. We believe this is partly due to the lack of support from switch vendors. It is an open question how the complexity of obtaining drop predictions and the complexity of push-out fare against each other. Although we focused on augmenting drop-tail algorithms with predictions,
we believe that our approach of using predictions has much potential also in other types of buffer algorithms.
While switch vendors may be better informed about the complexity of push-out buffers, an understanding of this complexity in the scientific community is much needed in order to navigate the complexity vs performance spectrum.

\subsection{Theory for Performance Guarantees}\label{sec:theory-future}
We believe the performance guarantees offered by \name can be improved in the future. Further, considering packet priorities and traffic classes in the competitive analysis is an open question.

\myitem{Improving consistency and robustness:} An open question is whether an algorithm could be designed to improve the competitive ratio under perfect predictions (consistency) better than $1.707$, while also improving the ratio under large error (robustness) better than $N$. Further research in this direction would enable a better understanding whether a consistency-robustness tradeoff exists for the buffer sharing problem.

\myitem{Competitive analysis with packet priorities:} Literature in theory considers that all packets are of same priority in the context of competitive analysis. However, it is well-known that preferential treatment of packets has various performance benefits especially in terms of flow completion times when short flow packets are prioritized. To this end, developing analysis techniques for such a setup is an interesting future research direction.

\section{Conclusion}
We presented \name, the first buffer sharing algorithm augmented with predictions that not only reaches close to optimal performance given low prediction error but also guarantees performance with arbitrarily large prediction error, while maintaining smoothness. We analytically proved our claims and our evaluations show the superior performance of \name even with an off-the-shelf machine-learned predictor, compared to the state-of-the-art buffer sharing algorithms. The building blocks required for \name are all individually practical in today's hardware. In future, we plan to pursue switch vendors to further discuss the integration of predictions with buffer sharing algorithm in hardware.

\section*{Acknowledgements}
\noindent
This work is part of a project that has received funding from the European Research Council (ERC) under the European Union’s Horizon 2020 research and innovation programme, consolidator project Self-Adjusting Networks (AdjustNet), grant agreement No. 864228, Horizon 2020, 2020-2025.
\begin{figure}[!h]
	\centering
	\includegraphics[width=0.6\linewidth]{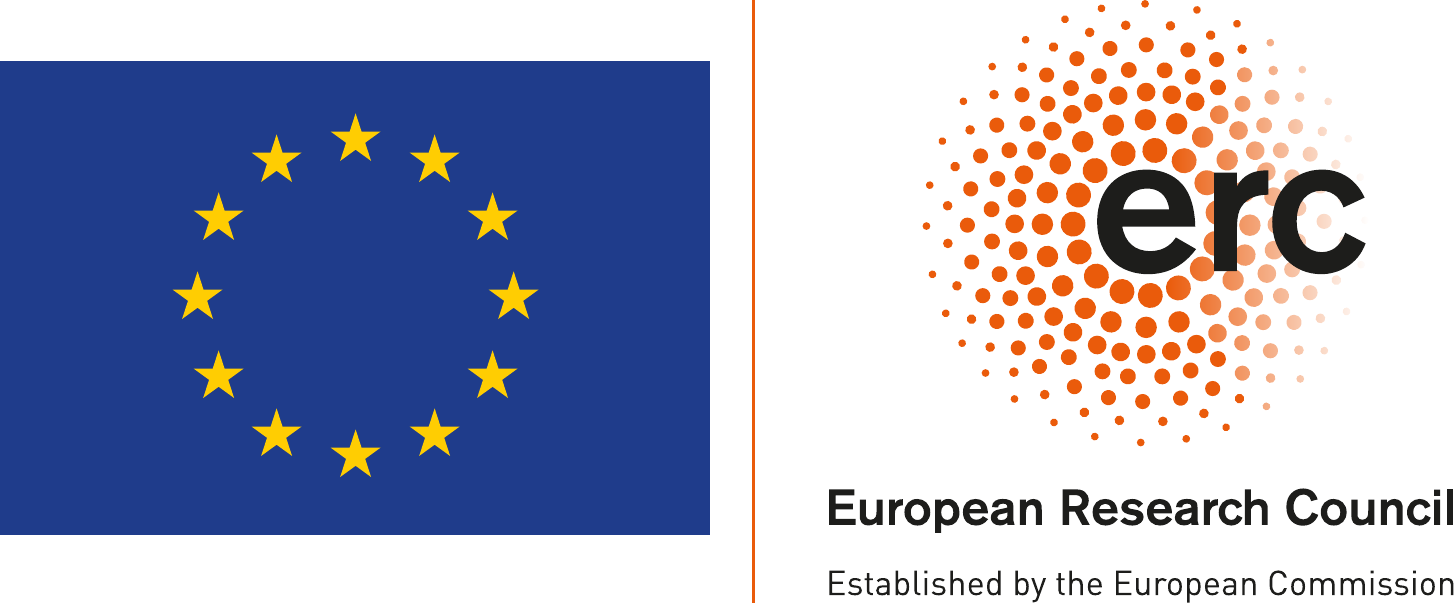}
\end{figure}

\label{bodyLastPage}

\bibliographystyle{plain}
\bibliography{references.bib}

\newpage
\appendix

\section{Model and Definitions}\label{app:model}
We consider a network switch equipped with an on-chip buffer size of $B$ units shared by $N$ ports. We mainly follow the widely used model in the literature~\cite{competitiveBuffer,breakingBarrier2,KESSELMAN2004161}. Time is discrete, and we refer to each step as timeslot. Packets (each of size unit $1$) arrive in an online manner as time progresses. Each timeslot is divided into two phases, arrival phase and departure phase. During each arrival phase, at most $N$ number of packets (in aggregate) arrive destined to $N$ ports. During each departure phase, every queue drains out one packet unless the queue is empty. A buffer sharing algorithm manages the shared buffer allocation across the $N$ ports. We next define preemptive (push-out) and non-preemptive (drop-tail) buffer sharing.

\begin{definition}[Preemptive buffer sharing]
	During every arrival phase, the buffer sharing algorithm is allowed to preempt \ie drop any number of existing packets in the buffer.
\end{definition}

\begin{definition}[Non-preemptive buffer sharing]
	During every arrival phase, the buffer sharing algorithm is only allowed to accept or drop the incoming packet. Every accepted packet must eventually be drained out from the corresponding queue.
\end{definition}

We denote by $\sigma(t) = (\sigma_i(t), \sigma_i(t),...,\sigma_N(t))$, an $N$-tuple, where $\sigma_i(t)$ denotes the number of packets arriving at time $t$ to queue $i$. We study the performance of a buffer sharing algorithm in terms of throughput \ie our objective is to maximize the total number of packets transmitted over the entire arrival sequence. We compare the performance of our online algorithms against an offline optimal algorithm, which has access to the entire arrival sequence at $t=0$ and has infinite computational capacity.

\begin{definition}[Competitive ratio]\label{def:competitive-ratio}
	Let ALG be an online algorithm and OPT be an offline optimal algorithm for the buffer sharing problem. Let ALG$(\sigma)$ and OPT$(\sigma)$ be the total number of packets transmitted by ALG and OPT for the arrival sequence $\sigma$. We say ALG is $c$-competitive if it satisfies the following condition for any arrival sequence $\sigma$.
	\begin{equation}
		OPT(\sigma) \le c\cdot ALG(\sigma)
	\end{equation}
\end{definition}

\section{FollowLQD: A Deterministic Algorithm}\label{app:detalg}
In this section we propose a new online deterministic algorithm FollowLQD in the non-preemptive case which is a non-predictive building block of \name. Intuitively, FollowLQD simply \emph{follows} the Longest Queue Drop (LQD) queues in the preemptive model. In particular, FollowLQD maintains a threshold $T_i(t)$ for each queue at time $t$. The thresholds are updated for every packet arrival and departure according to LQD in the preemptive model. We present the pseudocode for FollowLQD in Algorithm~\ref{alg:follow}. While FollowLQD tries to \emph{follow} LQD queue lengths by accepting packets as long as the queue lengths are smaller than the thresholds, it may happen that FollowLQD queues are larger than their thresholds. This is since FollowLQD cannot preempt (remove) existing packets in the buffer whereas LQD can preempt and correspondingly the thresholds may drop below the queue lengths. FollowLQD simply drops an incoming packet if it finds that the corresponding queue exceeds its threshold.

Although LQD is known to be $1.707$-competitive, we show that FollowLQD is still at least $\frac{N+1}{2}$-competitive. We present our lower bound based on a simple arrival sequence.

\begin{algorithm}[t]
	\SetKwFunction{arrival}{\textsc{\textcolor{red}{arrival}}}
	\SetKwFunction{updateThreshold}{\textsc{\textcolor{red}{updateThreshold}}}
	\SetKwFunction{departure}{\textsc{\textcolor{red}{departure}}}

	\SetKwProg{Fn}{function}{:}{}
	\SetKwProg{Proc}{procedure}{:}{}
	\SetKwInOut{KwIn}{Input}
	\SetKwInOut{KwOut}{Output}

	\KwIn{\ $\sigma(t)$ }

	\Proc{\arrival{$\sigma(t)$}}{

		\For{each packet $p \in \sigma(t)$}{
			Let $i$ be the destination queue for the packet $p$

			\textsc{updateThreshold}($i$, \emph{arrival})

			\eIf{$q_i(t) < T_i(t)$}{
				\If{$Q(t) < B$}{
					$q_i(t) \leftarrow q_i(t) + 1$ \Comment{accept}
				}
			}{
				\Comment{Drop}
			}
		}

	}

	\Proc{\departure{i}}{

		\If{$q_i(t) > 0$}{

			$q_i(t) \leftarrow q_i(t) - 1$ \Comment{Drain one packet}

		}

		\textsc{updateThreshold}($i$, \emph{departure})
	}

	\Fn{\updateThreshold{$i, event$}}{

		\If{event $=$ arrival}{
			\eIf(\Comment{Sum of thresholds}){$\Gamma(t) = B$}{

				Let $T_j(t)$ be the largest threshold

				$T_j(t) \leftarrow T_j(t) - 1$ \Comment{Decrease}

				$T_i(t) \leftarrow T_i(t) + 1$ \Comment{Increase}
			}{
				$T_i(t) \leftarrow T_i(t) + 1$ \Comment{Increase}

				$\Gamma(t) \leftarrow \Gamma(t) +1$
			}
		}
		\If{event $=$ departure}{
			\If{$T_i(t) > 0$}{
				$T_i(t) \leftarrow T_i(t) - 1$ \Comment{Decrease}

				$\Gamma(t) \leftarrow \Gamma(t) -1$
			}
		}
	}

	\caption{FollowLQD}
	\label{alg:follow}
\end{algorithm}

\begin{observation}
	FollowLQD is at least $\frac{N+1}{2}$-competitive.
\end{observation}
\begin{proof}
	We construct an arrival sequence such that for every two packets transmitted by FollowLQD, the offline optimal algorithm OPT transmits $N+1$ packets. Consider that all the queues are empty at time $t=0$. We then burst packets to a single queue say $i$ until its queue length reaches $B$. Note that this is possible since the threshold for queue $i$ that follows the corresponding LQD queue also grows up to $B$. At the end of the departure phase, FollowLQD transmits one packet and the queue length becomes $B-1$. At this point, we send $N$ packets, one packet to each of the $N$ queues. The thresholds are updated based on LQD, which has the following actions: \first preempt $N-1$ packets from queue $i$ and \second accept all $N$ packets to $N$ queues. Correspondingly, the threshold for queue $i$ of FollowLQD drops to $B-N+1$ but it still has $B-1$ packets in queue $i$. As a result, it can only accept one packet out of the $N$ incoming packets. At the end of the departure phase during this timeslot, FollowLQD has $B-1$ packets in queue $i$ and has transmitted $1$ packet in total. In the next timeslot, we send $N$ packets to the queue $i$ so that LQD's queue $i$ now gets back to size $B$ again. As the threshold is larger than the queue length ($B-1$), FollowLQD accepts $1$ packet. At the end of the departure phase, FollowLQD transmits $1$ packet from the queue $i$. Overall, FollowLQD transmitted $2$ packets but OPT transmitted $N+1$ packets. We then repeat the sequence such that for every $N+1$ packets transmitted by OPT, FollowLQD transmits $2$ packets. The competitive ratio is then at least $\frac{N+1}{2}$.
\end{proof}

\section{Buffer Sharing Algorithms with Predictions}\label{app:predictions}
In this section, we introduce our model for buffer sharing where there exists an oracle that predicts packet drop (or accept) for each packet in the arrival sequence $\sigma$, according to the prediction model introduced in \S\ref{sec:prediction-model}. We denote the drop sequence of LQD for the arrival sequence $\sigma$ by $\phi(\sigma)$, and the predicted drop sequence by $\phi^\prime(\sigma)$.
We classify prediction for each packet in to four types: true positive, false positive, true negative and false negative (see Figure~\ref{fig:predictions}. We denote the sequence of true positive predictions by $\phi^\prime_{TP}(\sigma)$, false positive predictions by $\phi^\prime_{FP}(\sigma)$, true negative predictions by $\phi^\prime_{TN}(\sigma)$ and false negative predictions by $\phi^\prime_{TP}(\sigma)$.
We drop $\sigma$ in our notations when the context is clear.

Hereafter, we mainly compare our online non-preemptive algorithm with predictions against online LQD (preemptive). We define the error made by the oracle by the following error function.

\errorFunction*

We now analyze \name that relies on drop predictions $\phi^\prime$ and takes decisions in pursuit of \emph{following} LQD more accurately. Algorithm~\ref{alg:follow-pred} presents the pseudocode for our algorithm.

In essence, perfect predictions allows us to perfectly follow LQD queues, essentially transmitting as many packets as LQD. However, we are also concerned about the performance of the algorithm when the oracle makes mispredictions.
In the following, we study the competitive ratio of \name as the error grows. We obtain the competitive ratio as a function of the error $\eta$: we show that \name is $1$-competitive against LQD with perfect predictions but at most $N$-competitive when the error is arbitrarily large.

\theoremCRatioPred*

\lemmaAlgLqd*

\begin{figure*}
	\centering
	\begin{subfigure}{0.248\linewidth}
		\centering
		\includegraphics[width=1\linewidth]{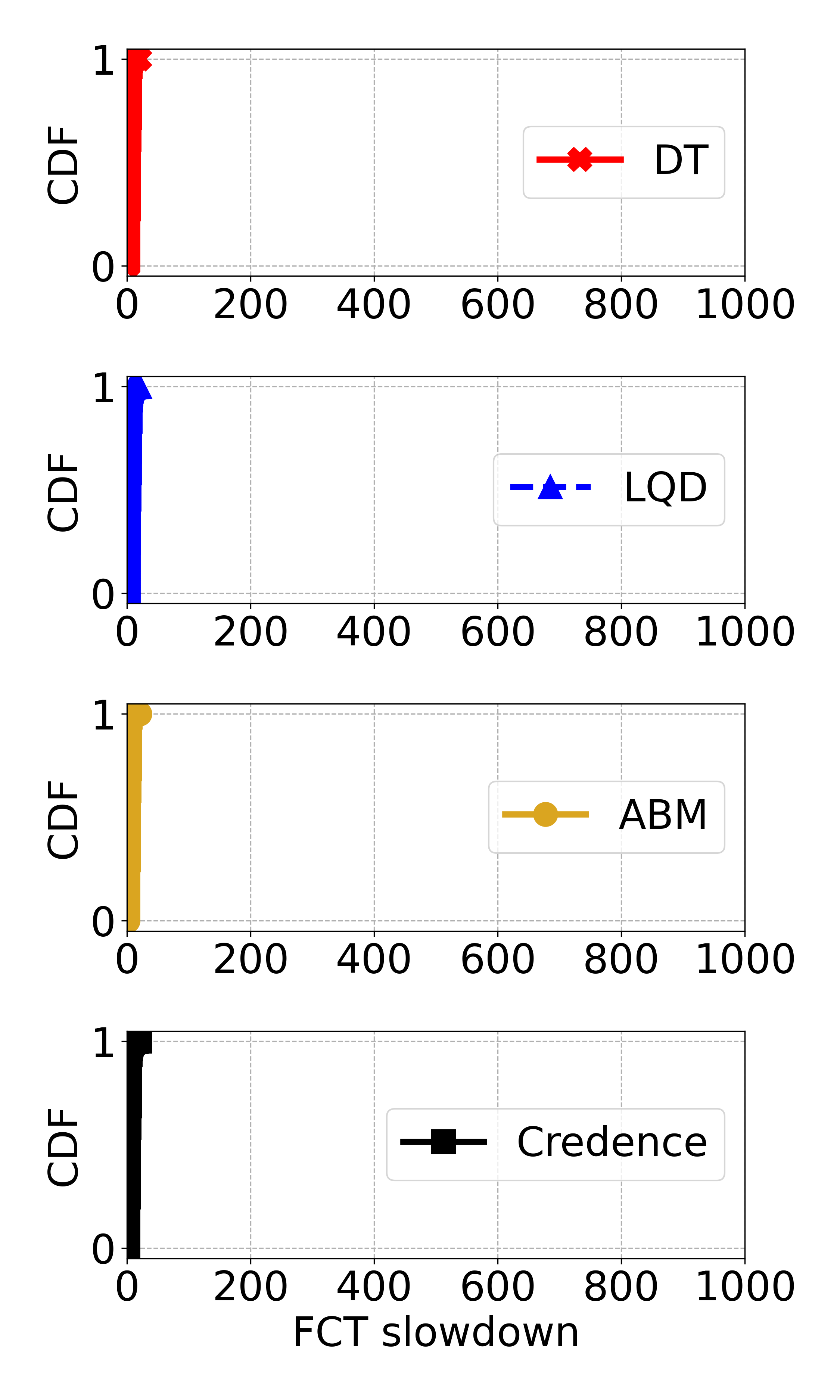}
		\caption{Burst size $=$ $12.5$\%}
		\label{}
	\end{subfigure}\hfill
	\begin{subfigure}{0.248\linewidth}
		\centering
		\includegraphics[width=1\linewidth]{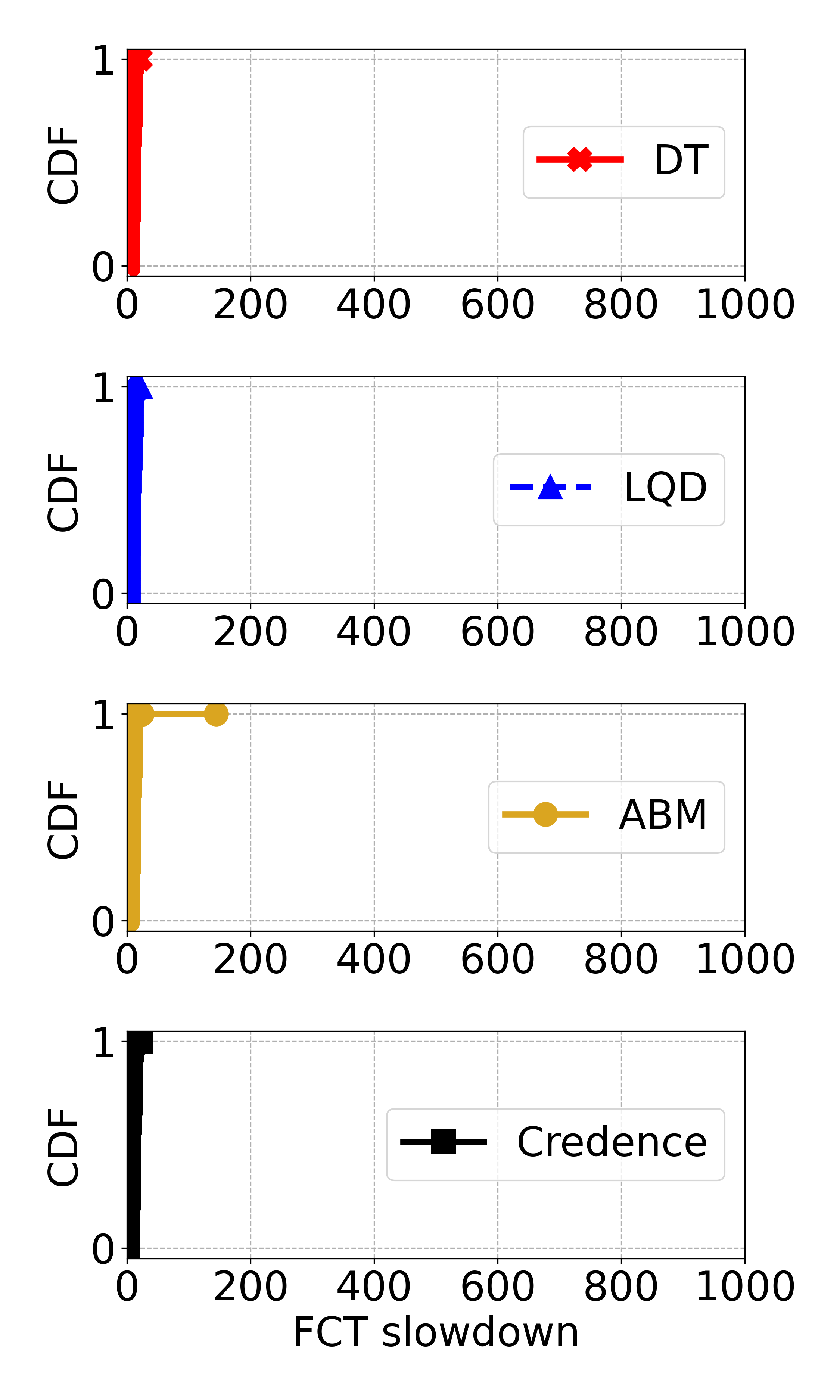}
		\caption{Burst size $=$ $25$\%}
		\label{}
	\end{subfigure}\hfill
	\begin{subfigure}{0.248\linewidth}
		\centering
		\includegraphics[width=1\linewidth]{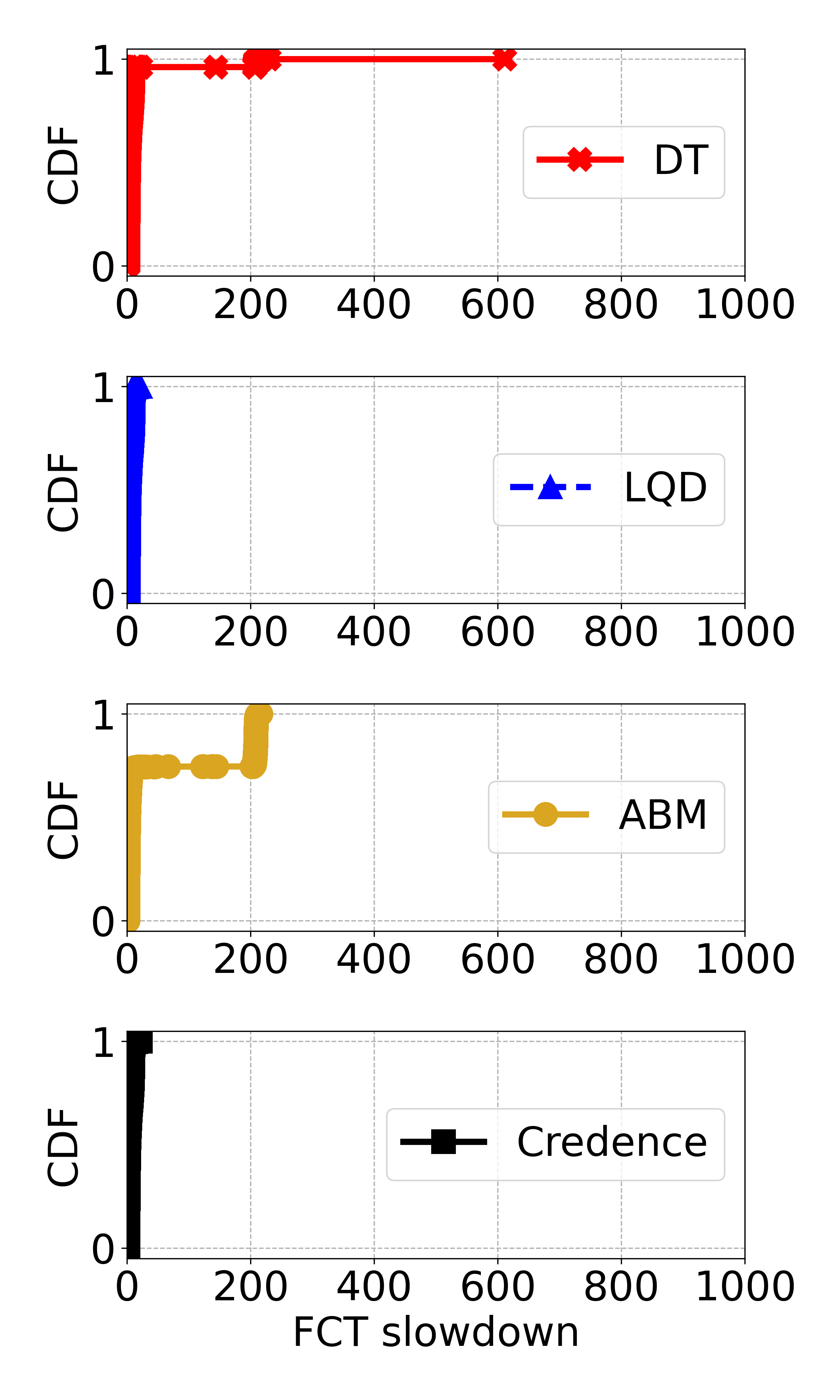}
		\caption{Burst size $=$ $50$\%}
		\label{}
	\end{subfigure}\hfill
	\begin{subfigure}{0.248\linewidth}
		\centering
		\includegraphics[width=1\linewidth]{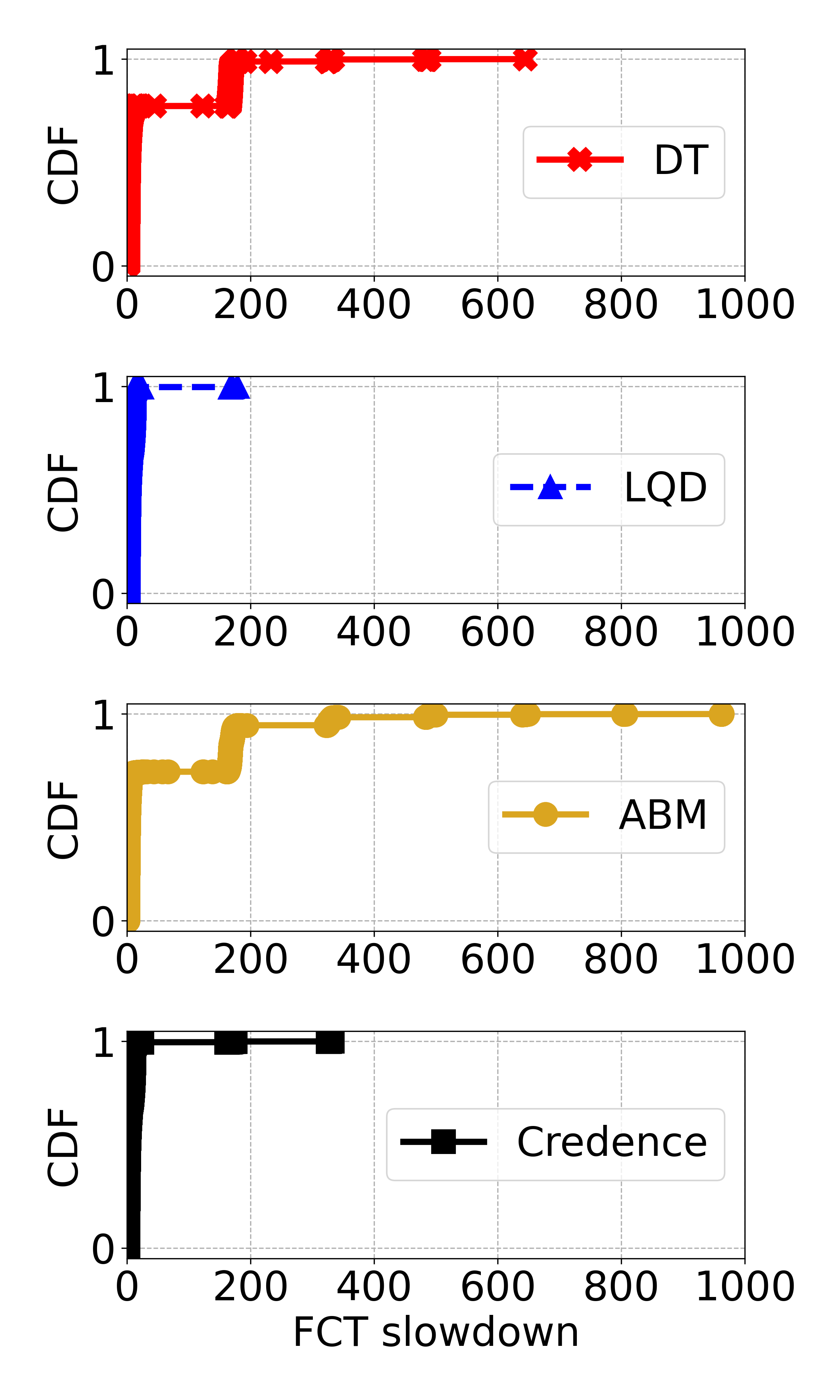}
		\caption{Burst size $=$ $75$\%}
		\label{}
	\end{subfigure}\hfill
	\caption{CDF of flow completion times (slowdown) for \name, DT, ABM and LQD across various burst sizes of incast workload and websearch workload at $40$\% load, with DCTCP as the transport protocol. Burst size is expressed as a percentage of the buffer size.}
	\label{fig:cdf-dctcp-burst}
\end{figure*}

\begin{proof}
	For simplicity, we refer to \name as ALG in the following.
	We prove our claim by analyzing the drops of ALG and relating the transmitted packets by $ALG(\sigma)$ to $LQD(\sigma^\prime)$. Every drop of $ALG$ arises from three types of situations.
	First, ALG can drop a packet due to the thresholds. Note that the thresholds used by $ALG$ correspond to the queue lengths of preemptive LQD over the same arrival sequence $\sigma$. As a result, both ALG and FollowLQD algorithm have the same thresholds at any time instance.
	Second, ALG drops a packet if the prediction is either true positive or false positive if and only if the queue length satisfies the corresponding thresholds. This type of drops are at most the total number of true positive and false positive predictions.
	Third, ALG drops a packet when the buffer is full which is the same condition for FollowLQD.
	In essence, ALG drops at most all the positive predictions and drops at most the number of packets dropped by FollowLQD serving the arrival sequence $\sigma-\phi^\prime_{TP}-\phi^\prime_{FP}$ \ie the arrival sequence in which all the packets predicted as positive are removed from $\sigma$.
	In order to prove our main claim, it remains to argue that the extra packets accepted by ALG due to the safeguard condition do not result in additional drops compared to FollowLQD with the arrival sequence $\sigma - \phi^\prime_{TP} - \phi^\prime_{FP}$. For every packet that fails to satisfy the threshold but gets accepted due to the safeguard condition by ALG, could cause at most one extra drop due to the thresholds before the buffer full again compared to FollowLQD. This is since, if FollowLQD accepts a packet, then its queue length is certainly less that the corresponding threshold (that is same for ALG). However, the queue length of ALG may have some extra packets that are accepted due to the safeguard condition. As a result, each such extra packet (dropped by FollowLQD) contributes to at most one drop compared compared to FollowLQD and the transmitted packets remains equivalent. Further, by the time the buffer is full in ALG, all the extra packets accepted due to the safeguard condition would have been drained out of the buffer. This is due to the fact that any such extra packet is at a queue length of at most $\frac{B}{N}$ (the safeguard condition) that drains out before the buffer fills up \ie it takes at least $\frac{B}{N}$ timeslots to fill the buffer (only $N$ packets can arrive in each timeslot).
	As a result, ALG transmits at least the total number of packets transmitted by FollowLQD over the arrival sequence  $\sigma-\phi^\prime_{TP}-\phi^\prime_{FP}$ \ie
	\[
		ALG(\sigma)\ge FollowLQD(\sigma-\phi^\prime_{TP}-\phi^\prime_{FP})
	\]
	Using Definition~\ref{def:error-function}, we express FollowLQD in terms of LQD and the error function $\eta(\phi,\phi^\prime)$, and obtain Equation~\ref{eq:alg-lqd}.
\end{proof}

\lemmaAlgOpt*

\begin{proof}
	Irrespective of the predictions, \name always accepts an incoming packet if the longest queue is less than or equal to $\frac{B}{N}$. When \name drops a packet, there is at least one queue that has at least $\frac{B}{N}$ number of packets. Hence, every packet in OPT can be matched to at least $\frac{B}{N}$ number of packets. Consequently, the competitive ratio is at most $N$.
\end{proof}

We are now ready to prove our main claim (Theorem~\ref{theorem:c-ratio-pred}) using the above results.

\begin{proof}[\textbf{Proof of Theorem~\ref{theorem:c-ratio-pred}}]
	From Definition~\ref{def:competitive-ratio}, in order to prove the competitive ratio of our \name, we are mainly concerned with the upper bound of $\frac{OPT(\sigma)}{\name(\sigma)}$ for any arrival sequence $\sigma$. Since $\frac{OPT(\sigma)}{LQD(\sigma)}\le 1.707$ is known from literature~\cite{competitiveBuffer,breakingBarrier2}, we use this result to compare \name and LQD in order to argue about the competitive ratio i.e, $\frac{OPT(\sigma)}{\name(\sigma)} \le 1.707\cdot \frac{LQD(\sigma)}{\name(\sigma)}$ for any request sequence $\sigma$.
	From Lemma~\ref{lemma:alg-lqd}, we have the following:
	\[
		\frac{LQD(\sigma)}{\name(\sigma)} \le \eta(\phi, \phi^\prime)
	\]
	From Lemma~\ref{lemma:alg-opt}, irrespective of the predicted sequence, we have that $\frac{OPT(\sigma)}{\name(\sigma)}\le N$. Finally, since $\frac{OPT(\sigma)}{\name(\sigma)}\le 1.707\cdot \frac{LQD(\sigma)}{\name(\sigma)}$, the competitive ratio of \name is given by:
	\[
		\frac{OPT(\sigma)}{\name(\sigma)} \le \min\left( 1.707\ \eta(\phi, \phi^\prime),\ N \right)
	\]
	The proof follows by Definition~\ref{def:competitive-ratio}.
\end{proof}

\begin{figure*}
	\centering
	\begin{subfigure}{0.248\linewidth}
		\centering
		\includegraphics[width=1\linewidth]{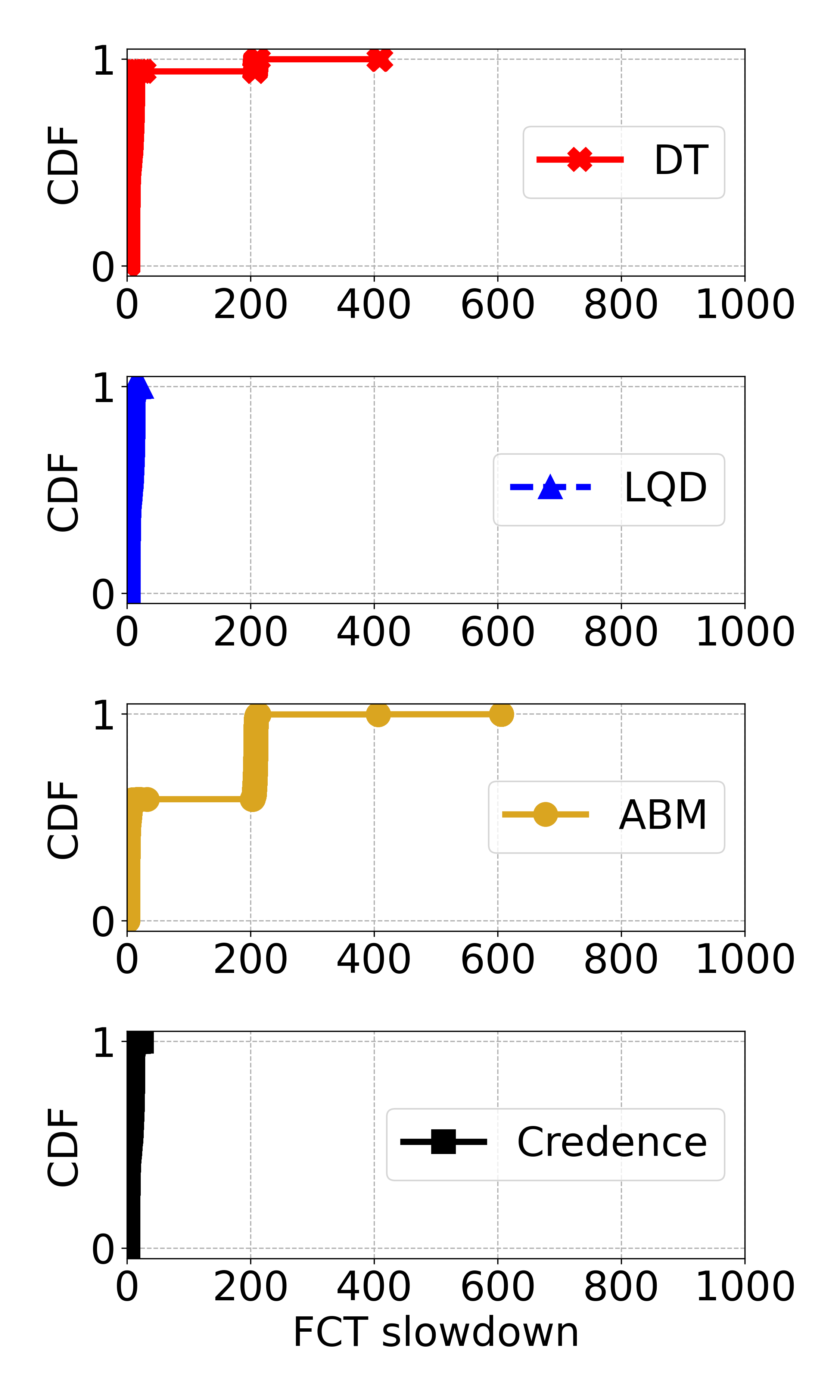}
		\caption{Load $=$ $20$\%}
		\label{}
	\end{subfigure}\hfill
	\begin{subfigure}{0.248\linewidth}
		\centering
		\includegraphics[width=1\linewidth]{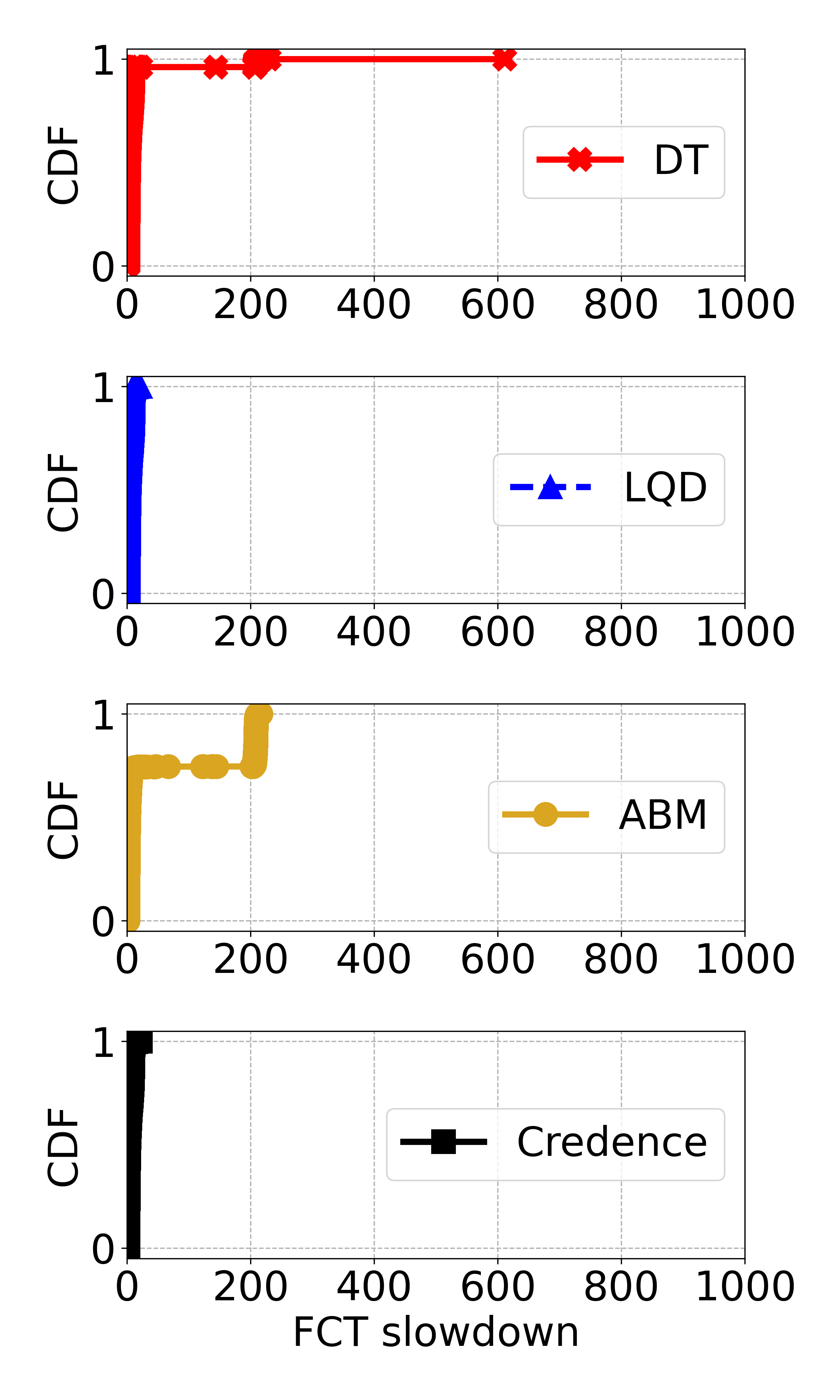}
		\caption{Load $=$ $40$\%}
		\label{}
	\end{subfigure}\hfill
	\begin{subfigure}{0.248\linewidth}
		\centering
		\includegraphics[width=1\linewidth]{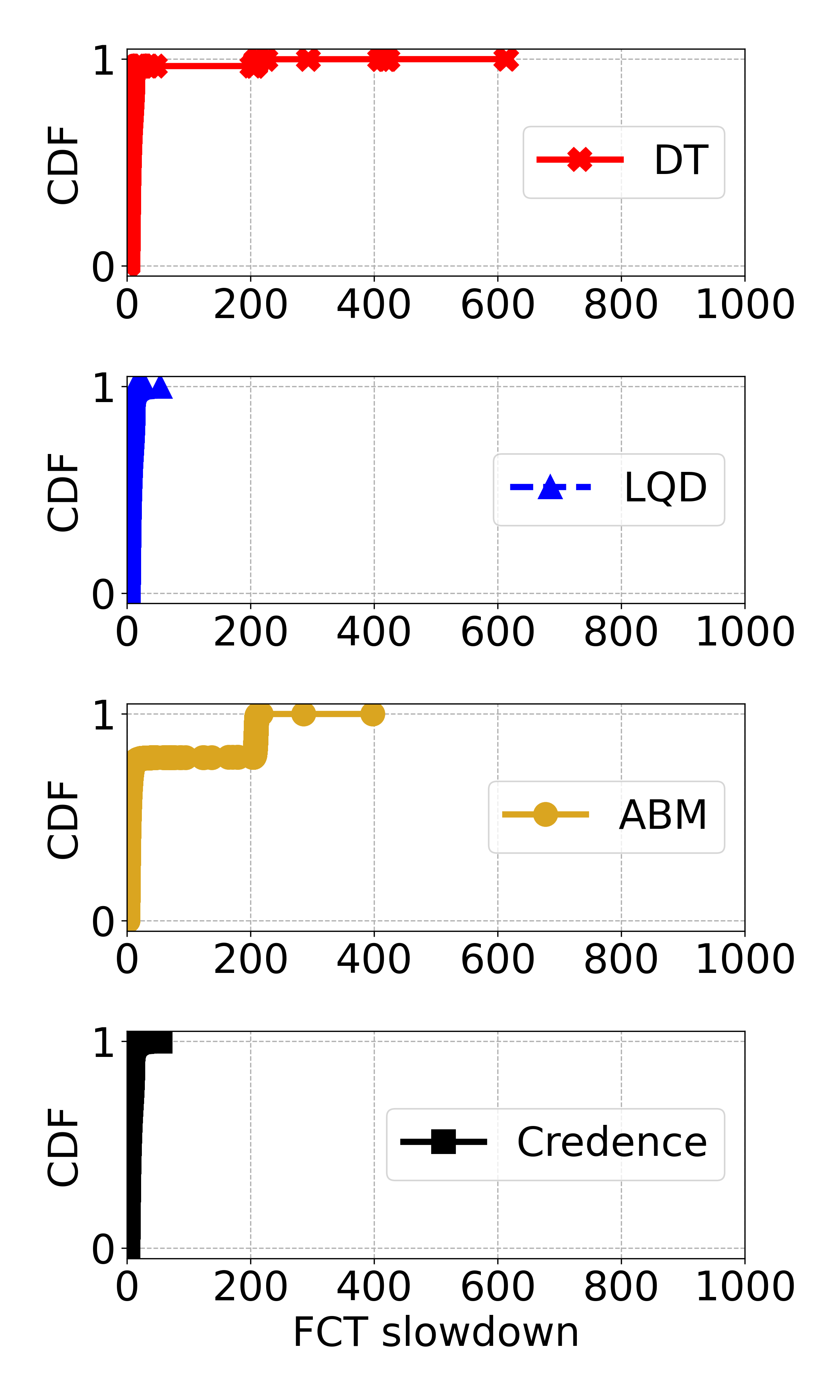}
		\caption{Load $=$ $60$\%}
		\label{}
	\end{subfigure}\hfill
	\begin{subfigure}{0.248\linewidth}
		\centering
		\includegraphics[width=1\linewidth]{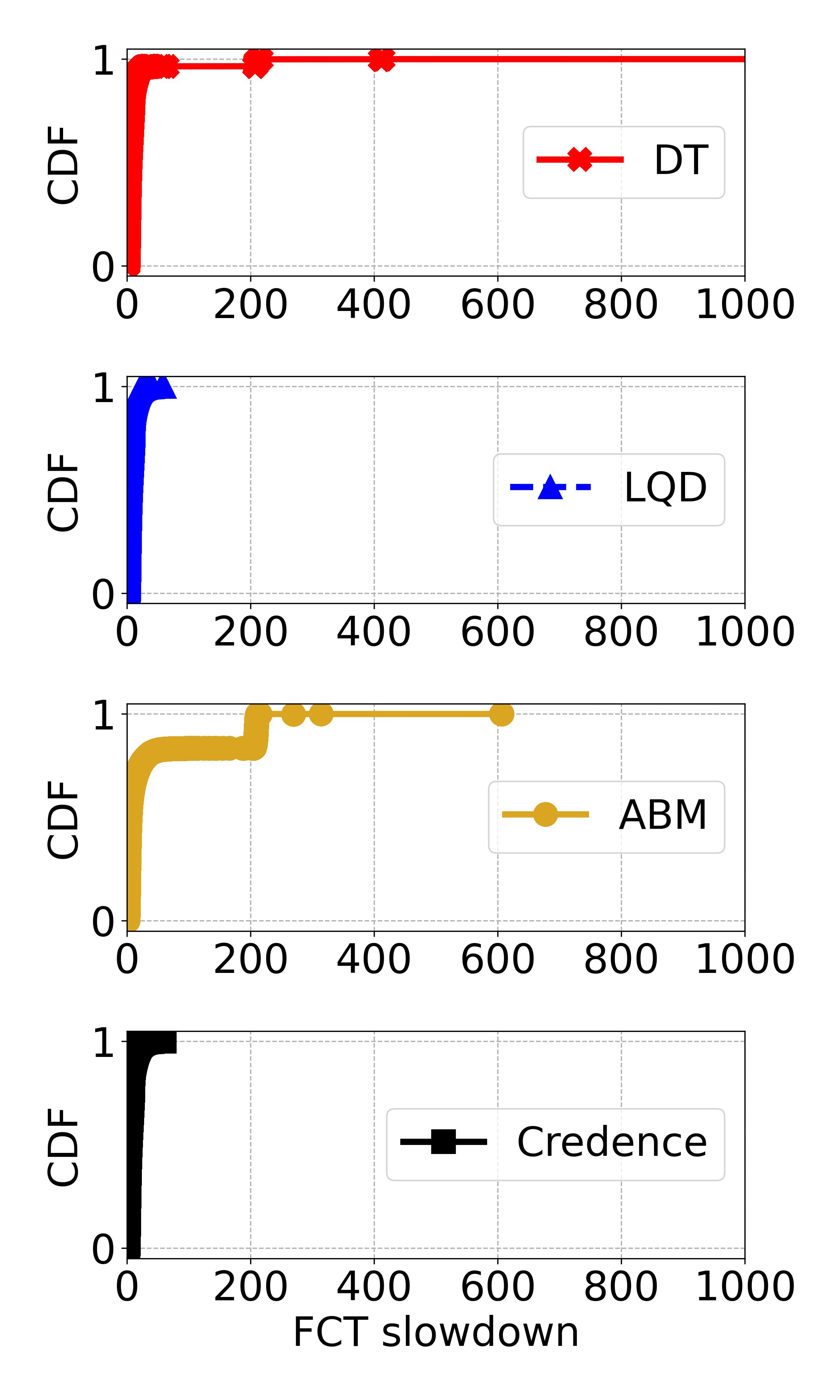}
		\caption{Load $=$ $80$\%}
		\label{}
	\end{subfigure}\hfill
	\caption{CDF of flow completion times (slowdown) for \name, DT, ABM and LQD across various loads of websearch workload and incast workload at a burst size $50$\% of the buffer size, with DCTCP as the transport protocol.}
	\label{fig:cdf-dctcp-load}
\end{figure*}

\begin{theorem}\label{th:error-upper}
	Let $\phi^\prime$ denote the sequence of drops predicted by the machine-learned oracle. Let $\phi_{TP}^\prime$, $\phi_{FP}^\prime$, $\phi_{TN}^\prime$, and $\phi_{FN}^\prime$ denote the sequence of true positive, false positive, true negative and false negative predictions for the arrival sequence $\sigma$.
	The error function $\eta(\phi,\phi^\prime)$ (Definition~\ref{def:error-function}) is upper upper bounded as follows:
	\[
		\eta(\phi,\phi^\prime) \le \frac{\phi^\prime_{TN} + \phi^\prime_{FP}}{\displaystyle \phi^\prime_{TN} - \displaystyle\min\left((N-1)\cdot\phi^\prime_{FN}, \phi^\prime_{TN}\right)}
	\]
\end{theorem}

\begin{proof}
	Our proof is based on two arguments: \first $LQD(\sigma) = \phi^\prime_{TN} + \phi^\prime_{FP}$ and \second $FollowLQD(\sigma-\phi^\prime_{TP}-\phi^\prime_{FP}) \ge \phi^\prime_{TN} - \displaystyle\min\left((N-1)\cdot\phi^\prime_{FN}, \phi^\prime_{TN}\right)$.

	First, $LQD(\sigma)$ is the total number of transmitted packets by LQD. Recall that the ground-truth (transmitted by LQD) for a prediction is an accept if and only if the prediction is either true negative or a false positive. Hence, the total number of packets transmitted by LQD \ie $LQD(\sigma)$ is the sum of true negative predictions and the false positive predictions.

	Second, $FollowLQD(\sigma - \phi^\prime_{TP} - \phi^\prime_{FP})$ transmits at least $\phi^\prime_{TN} + \phi^\prime_{FN} - Y$, where $Y$ is the total number of drops caused by false negative predictions. Note that $\sigma = \phi^\prime_{TN} + \phi^\prime_{FN} + \phi^\prime_{FP} + \phi^\prime_{TP}$. The proof follows by showing that each false negative results in at most $N$ extra drops due to the buffer limit. Further, these extra drops \emph{must} be true negative predictions since we have already removed positive predictions from our arrival sequences (\ie we assume at most all the positive predictions have be dropped). Additionally, since the extra drops are true negative predictions, it implies that LQD transmits those packets but although our prediction is true, we incur additional drop due to the buffer limit. For each false negative, there can be at most one drop in a single timeslot for up to $N-1$ distinct timeslots such drops that LQD accepts and transmits those packets but FollowLQD drops them. Beyond $N-1$ drops, there can only be at most $1$ other drop upon which the existence of an additional packet (false negative) in FollowLQD's buffer would be nullified. This is since, during the initial $N-1$ drops, FollowLQD could not accept the incoming packet but after the transmission phase, the queues having false negative predictions decrement their size by $1$. This leaves at least $N-1$ packets free space in FollowLQD after $N-1$ drops and LQD also has the same remaining space after those extra $N-1$ accepted by LQD are also transmitted. At this time, both LQD and FollowLQD have the same remaining space and they also transmit the same number of packets in each timeslot. One additional drop by FollowLQD corresponding to a false negative is still possible due the thresholds \ie if there exists a packet arrival to the queue having false negative, the incoming packet is dropped since the existence of false negatives implies that the queue length is large than the threshold. As a result, there are at most $N$ drops by FollowLQD for each false negative prediction.

	The proof follows by the above two arguments.
\end{proof}

For completeness, although well-known in the literature, we define accuracy, precision, recall and f1 score below (used in Figure~\ref{fig:scores} in \S\ref{sec:evaluation}).

\[
	Accuracy = \frac{\phi^\prime_{TP} + \phi^\prime_{TN}}{\phi^\prime_{TP} + \phi^\prime_{TN} + \phi^\prime_{FP} + \phi^\prime_{FN}}
\]
\[
	Precision = \frac{\phi^\prime_{TP}}{\phi^\prime_{TP} + \phi^\prime_{FP}}
\]
\[
	Recall = \frac{\phi^\prime_{TP}}{\phi^\prime_{TP} + \phi^\prime_{FN}}
\]
\[
	F1\ score = \frac{2\cdot \phi^\prime_{TP}}{2\cdot \phi^\prime_{TP} + \phi^\prime_{FP} + \phi^\prime_{FN}}
\]

\begin{figure*}
	\centering
	\begin{subfigure}{0.248\linewidth}
		\centering
		\includegraphics[width=1\linewidth]{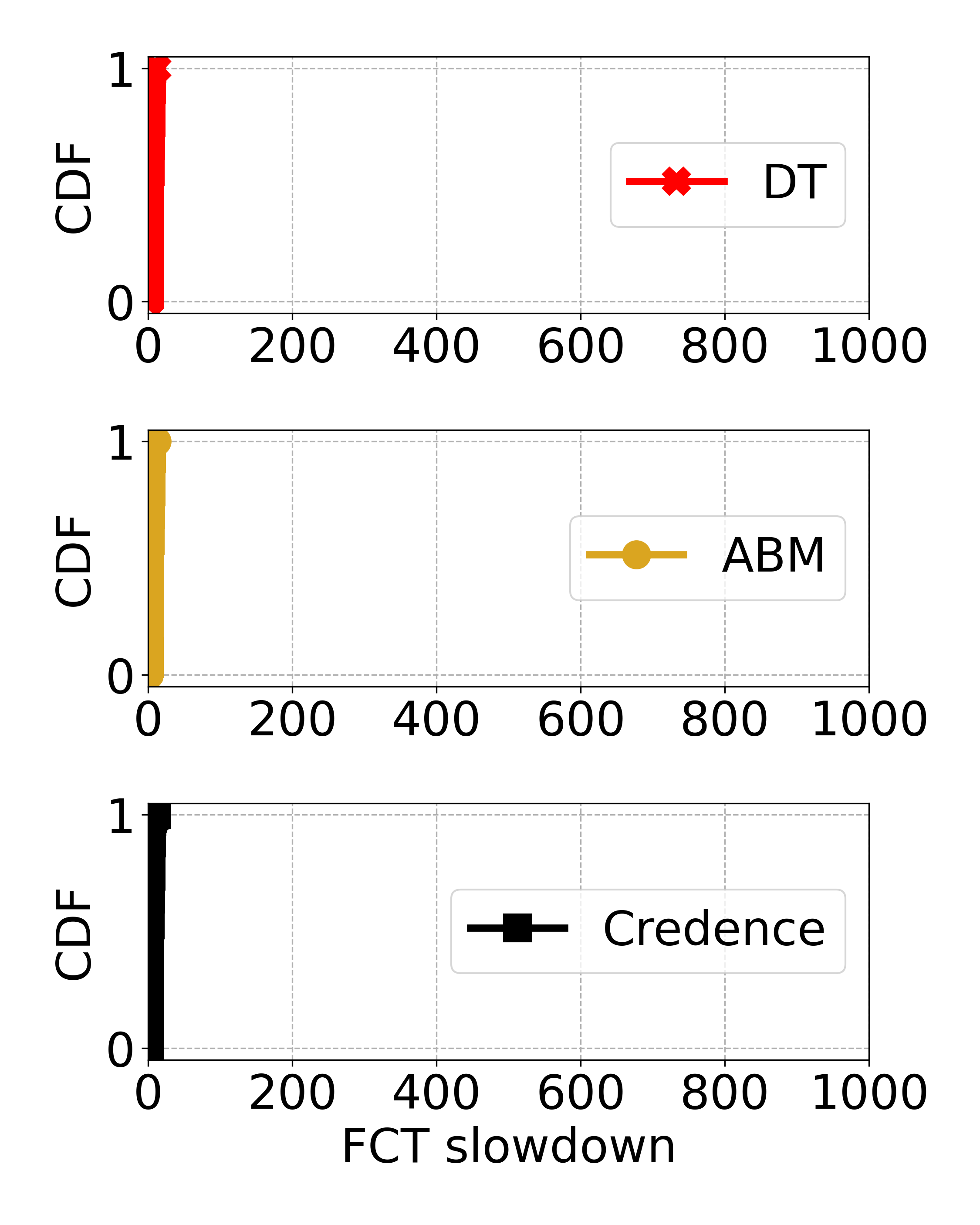}
		\caption{Burst size $=$ $12.5$\%}
		\label{}
	\end{subfigure}\hfill
	\begin{subfigure}{0.248\linewidth}
		\centering
		\includegraphics[width=1\linewidth]{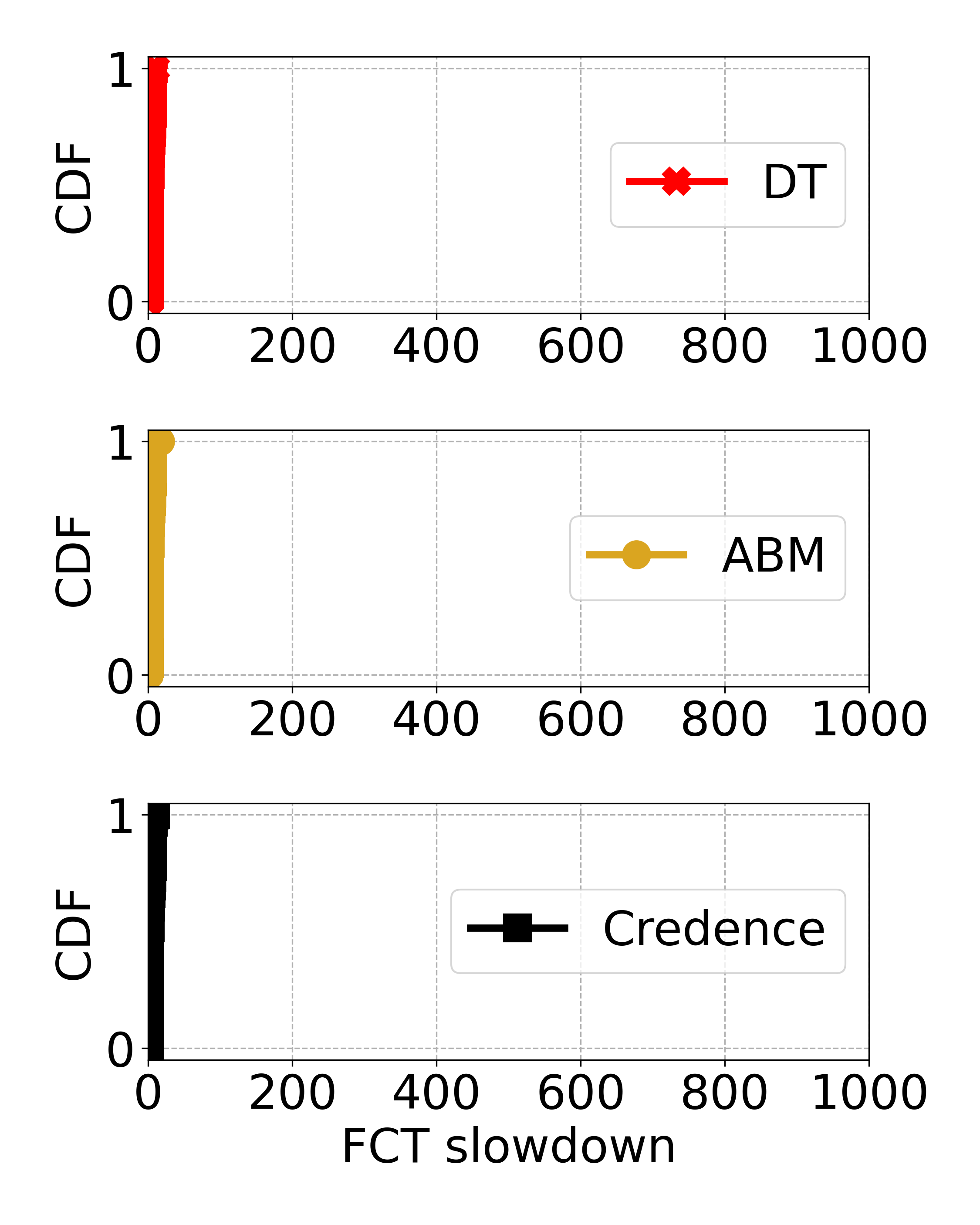}
		\caption{Burst size $=$ $25$\%}
		\label{}
	\end{subfigure}\hfill
	\begin{subfigure}{0.248\linewidth}
		\centering
		\includegraphics[width=1\linewidth]{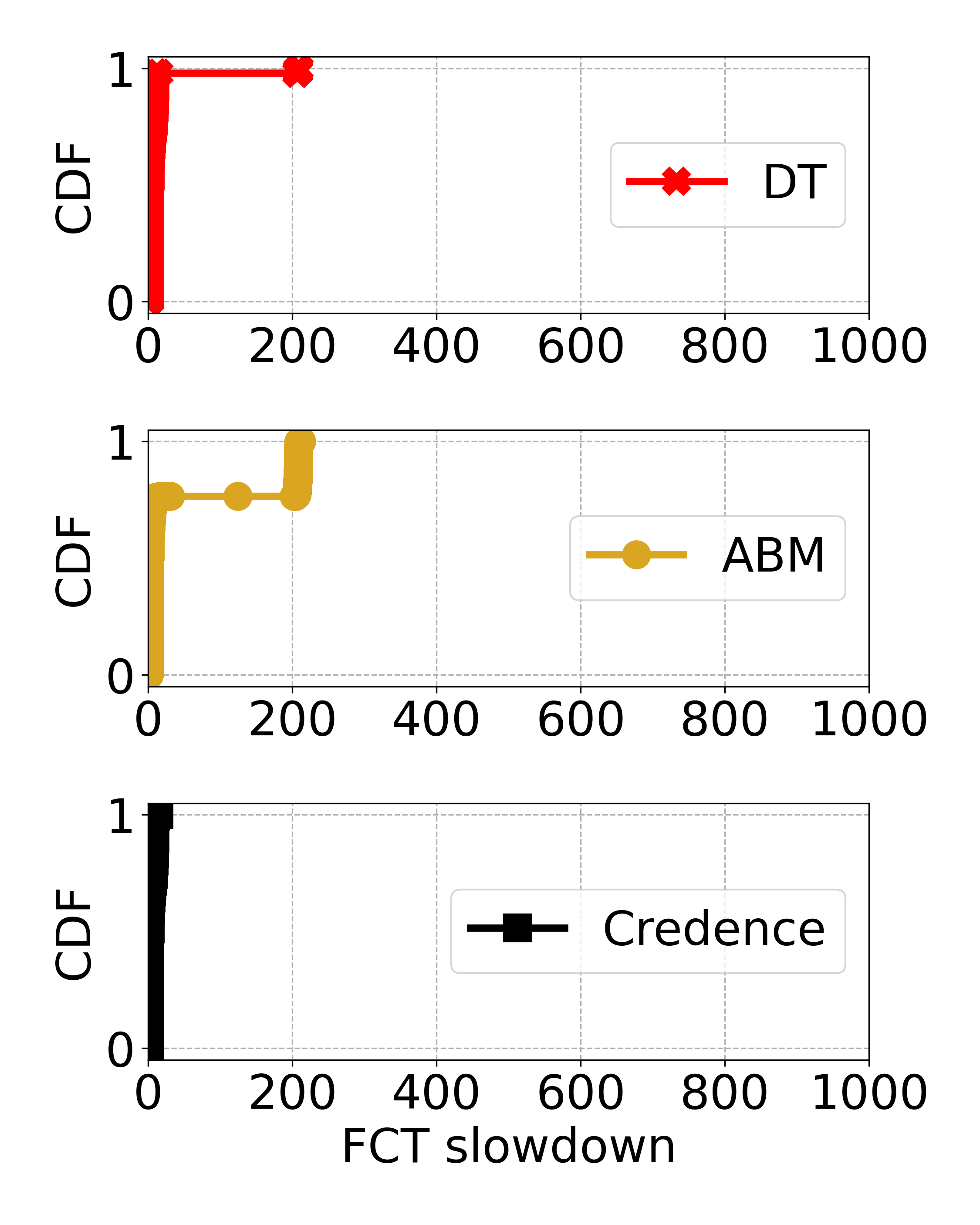}
		\caption{Burst size $=$ $50$\%}
		\label{}
	\end{subfigure}\hfill
	\begin{subfigure}{0.248\linewidth}
		\centering
		\includegraphics[width=1\linewidth]{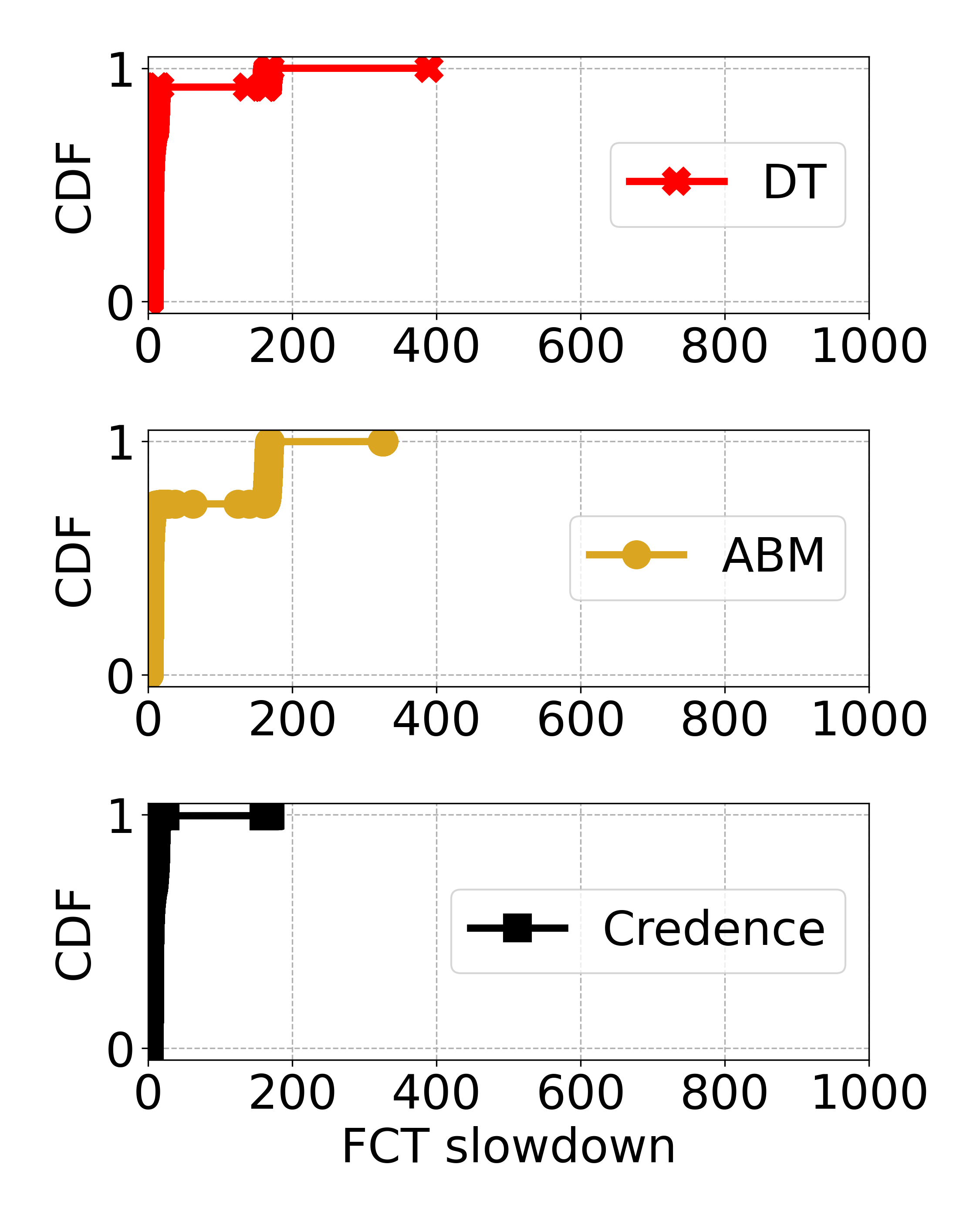}
		\caption{Burst size $=$ $75$\%}
		\label{}
	\end{subfigure}\hfill
	\caption{CDF of flow completion times (slowdown) for \name, DT, ABM and LQD across various burst sizes of incast workload and websearch workload at $40$\% load, with PowerTCP as the transport protocol. Burst size is expressed as a percentage of the buffer size.}
	\label{fig:cdf-powert-burst}
	\vspace{-5mm}
\end{figure*}

\section{Additional Results}\label{app:results}
In this section, we present additional results from our evaluations.
Figures~\ref{fig:cdf-dctcp-burst},~\ref{fig:cdf-dctcp-load},~\ref{fig:cdf-powert-burst} present the CDF of flow completion times for each experiment in our evaluations (\S\ref{sec:evaluation}), showing the complete performance profile of each algorithm.

Figure~\ref{fig:competitive-example} presents our numerical results based on a custom simulator in discrete time. Note that Figure~\ref{fig:competitive-example} shows the throughput \emph{ratio} of an algorithm vs LQD. We perform this experiment using custom simulator in order to fully control the prediction error (artificially).

We generate large bursts of the size of the total buffer, where each such burst arrives according to a poisson process (which is fixed in subsequent runs). We then collect a trace of per-packet drop (or accept) trace using LQD as the buffer sharing algorithm. This trace serves as the ground-truth as well as the case for perfect predictions for \name. We then run \name over the same packet arrival sequence from above, and use the drop trace of LQD as predictions. With full access to this trace \ie perfect predictions case, \name performs exactly as LQD as expected. However, in order to study the performance of \name with increasing error, in a controlled manner, we flip each packet drop (or accept) from our LQD's drop trace \ie each flip becomes a false prediction. We control the error via the flipping probability \ie the false prediction rate. We observe from Figure~\ref{fig:competitive-example} that \name degrade in throughput as the probability of false predictions increases \ie as the prediction error increases. However, even at as high as $0.7$ probability of false predictions, \name still out-performs DT.

\begin{figure}[!h]
	\centering
	\begin{minipage}{1\linewidth}
		\centering
		\includegraphics[width=0.7\linewidth]{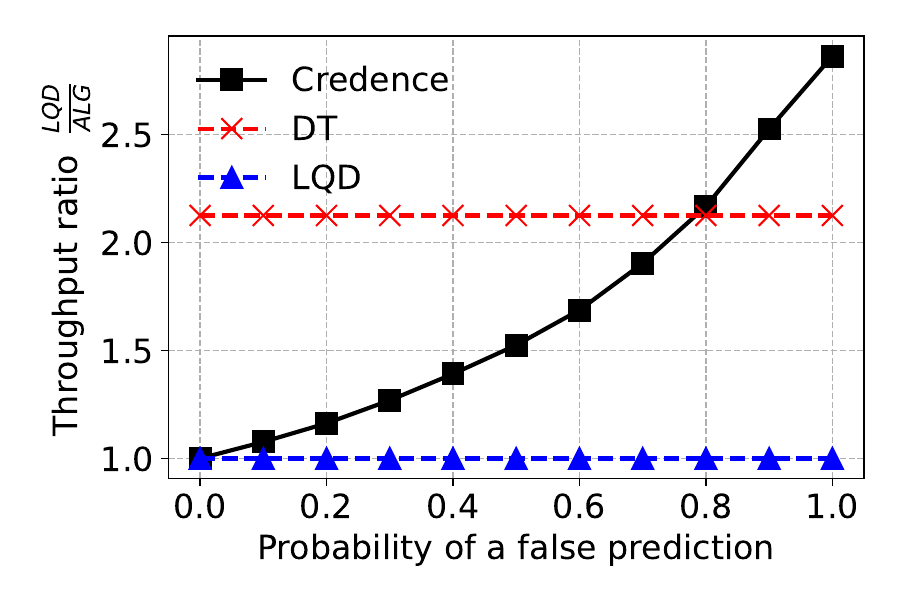}
		\caption{As the probability of false predictions increases, \name's throughput compared to LQD (push-out) \ie the ratio $\frac{LQD}{ALG}$ increases from $1$ to $2.9$ (lower values are better). \name performs significantly better than DT even when the probability of false predictions is as high as $0.7$.}
		\label{fig:competitive-example}
	\end{minipage}\hfill
	\begin{minipage}{1\linewidth}
		\centering
		\includegraphics[width=0.7\linewidth]{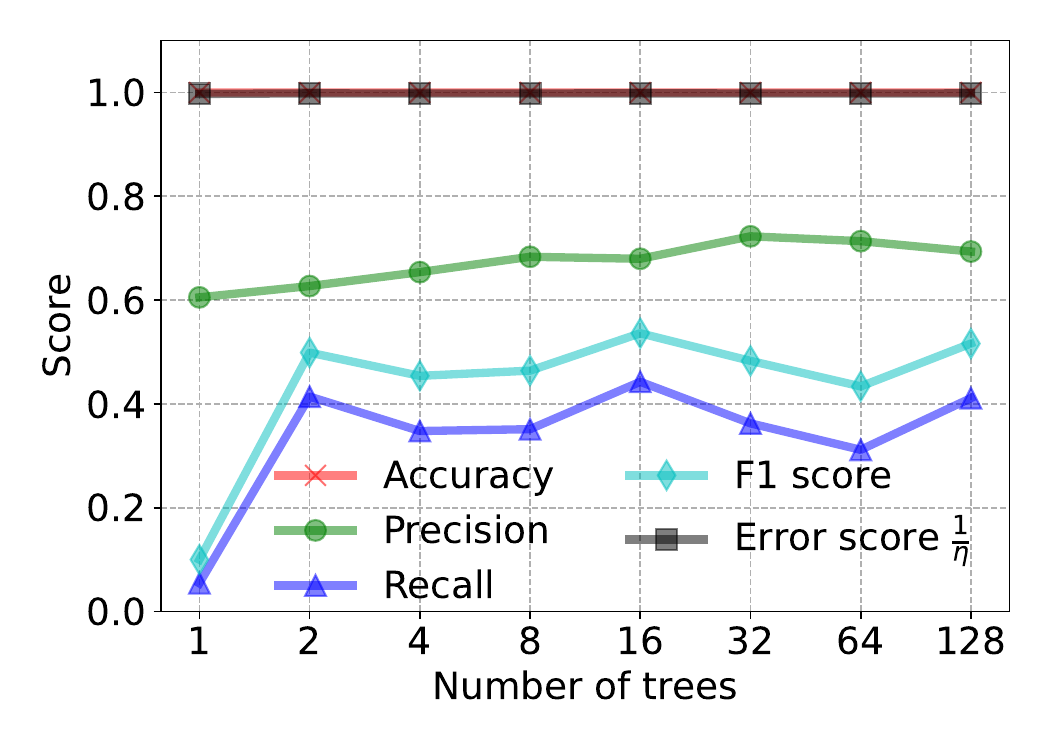}
		\vspace{-2mm}
		\caption{The quality of our predictions does not improve significantly beyond $4$ trees in our random forest classifier.\\ \\ \\ }
		\label{fig:scores}
	\end{minipage}
	\vspace{-5mm}
\end{figure}

In Figure~\ref{fig:scores}, we present our results obtained from a parameter sweep across the number of trees used for random forest model vs prediction~scores.

\label{LastPage}

\end{document}